\documentclass[11pt]{article}
\usepackage[margin=1in]{geometry}

\usepackage{authblk}
\usepackage[title]{appendix}
\usepackage{amsmath, amssymb, amsthm}
\usepackage{stmaryrd}
\usepackage{braket}
\usepackage{algorithm}
\usepackage{enumitem}
\usepackage{algpseudocode}
\newcommand{\sgn}{\operatorname{sgn}}
\DeclareMathOperator*{\E}{\mathbb{E}}
\DeclareMathOperator*{\spn}{\operatorname{span}}
\DeclareMathOperator*{\poly}{\operatorname{poly}}
\newcommand{\norm}[1]{\left\lVert#1\right\rVert}
\usepackage{etoolbox}
\usepackage{hyperref}
\apptocmd{\thebibliography}{\raggedright}{}{}
\usepackage{cite,color,float}
\usepackage{mdframed} 
\usepackage{cleveref}
\crefname{protocol}{protocol}{protocols}
\Crefname{protocol}{Protocol}{Protocols}
\crefname{thm}{theorem}{theorems}
\Crefname{thm}{Theorem}{Theorems}
\crefname{rmk}{remark}{remarks}
\Crefname{rmk}{Remark}{Remarks}
\crefname{lem}{lemma}{lemmata}
\Crefname{lem}{Lemma}{Lemmata}

\numberwithin{equation}{section}
\newcounter{protocol}
\newcommand{\linefill}{\rule{\linewidth}{0.8pt}}

\newenvironment{protocol}[1]{\begingroup\setlength\parindent{0pt}\medskip\noindent\linefill\\
\refstepcounter{protocol}\textbf{Protocol \theprotocol} #1\\\noindent\linefill}
{\vspace{-\topsep}\noindent\linefill\endgroup}

\def \Hin {H_{\mathrm{in}}}
\def \Hout {H_{\mathrm{out}}}
\def \Hprop {H_{\mathrm{prop}}}
\def \Hclock {H_{\mathrm{clock}}}
\def \Jclock {J_{\mathrm{clock}}}
\def \Jprop {J_{\mathrm{prop}}}
\def \Kin {K_{\mathrm{in}}}
\def \Kclock {K_{\mathrm{clock}}}
\def \Kprop {K_{\mathrm{prop}}}

\newcommand{\histpsi}[1]{\ket{\psi_{#1}^{\mathrm{hist}}}}
\newcommand{\LHXZ}[1]{\mathrm{LH}_{\mathrm{XZ}}^{#1}}
\newcommand{\ground}[1]{{\lambda \left (#1 \right)}}

\title{Constant-round Blind Classical Verification of Quantum Sampling}

\author[1]{Kai-Min Chung\thanks{\href{mailto:kmchung@iis.sinica.edu.tw}{kmchung@iis.sinica.edu.tw}. Partially supported by the 2019 Academia Sinica Career Development Award under Grant no. 23-17, and MOST QC project under Grant no. MOST 108-2627-E-002-001.}}
\author[2]{Yi Lee\thanks{\href{mailto:ylee1228@umd.edu}{ylee1228@umd.edu}. This work was done while affiliated to Academia Sinica and to National Taiwan University.}}
\author[3]{Han-Hsuan Lin\thanks{\href{mailto:linhh@cs.nthu.edu.tw}{linhh@cs.nthu.edu.tw}. Part of this work was done while supported by Scott Aaronson's Vannevar Bush Faculty Fellowship from the US Department of Defense. Partially funded by MOST Grant no. 110-2222-E-007-002-MY3}}
\author[4]{Xiaodi Wu\thanks{\href{mailto:xwu@cs.umd.edu}{xwu@cs.umd.edu}. Partially supported by the U.S. National Science Foundation grant CCF-1755800, CCF-1816695, and CCF-1942837 (CAREER). }}
\affil[1]{Institute of Information Science, Academia Sinica, Taiwan}
\affil[2]{Department of Computer Science, University of Maryland, USA}
\affil[3]{Department of Computer Science, National Tsing Hua University, Taiwan}
\affil[4]{
	Department of Computer Science, Institute for Advanced Computer Studies,
	and Joint Center for Quantum Information and Computer Science,
	University of Maryland, USA
}

\usepackage{graphicx,amsmath, amssymb,color,url,booktabs,comment}  
\newcommand{\nc}{\newcommand}
\nc{\rnc}{\renewcommand}

\def\View{\mathsf{View}}

\def\GS{\mathsf{Ham}}
\nc{\cVGS}{\ensuremath{\cV_\GS}}

\def\Ver{\mathsf{Ver}}
\nc{\cVVer}{\ensuremath{\cV_\Ver}}

\def\Samp{\mathsf{Samp}}
\nc{\PiSamp}{\ensuremath{\Pi_\Samp}}
\nc{\VSamp}{\ensuremath{V_\Samp}}
\nc{\PSamp}{\ensuremath{P_\Samp}}
\nc{\PSampstar}{\ensuremath{P_\Samp^*}}
\nc{\cVSamp}[1]{\ensuremath{\cV_{\Samp,#1}}}
\nc{\cPSamp}[1]{\ensuremath{\cP_{\Samp,#1}}}

\def\SampZ{\mathsf{Final}}
\nc{\PiSampZ}{\ensuremath{\Pi_\SampZ}}
\nc{\VSampZ}{\ensuremath{V_\SampZ}}
\nc{\PSampZ}{\ensuremath{P_\SampZ}}
\nc{\PSampZstar}{\ensuremath{P_\SampZ^*}}
\nc{\cVSampZ}[1]{\ensuremath{\cV_{\SampZ,#1}}}
\nc{\cPSampZ}[1]{\ensuremath{\cP_{\SampZ,#1}}}

\def\Rej{\mathsf{Rej}}

\def\QHE{\mathsf{QHE}}
\def\QGen{\mathsf{QHE.Keygen}}
\def\QEnc{\mathsf{QHE.Enc}}
\def\QEval{\mathsf{QHE.Eval}}
\def\QDec{\mathsf{QHE.Dec}}

\def\blind{\mathsf{blind}}
\nc{\Piblind}{\ensuremath{\Pi_\blind}}
\nc{\Vblind}{\ensuremath{V_\blind}}
\nc{\Pblind}{\ensuremath{P_\blind}}
\nc{\Pblindstar}{\ensuremath{P_\blind^*}}
\nc{\cVblind}[1]{\ensuremath{\cV_{\blind,#1}}}
\nc{\cPblind}[1]{\ensuremath{\cP_{\blind,#1}}}

\def\Pstar{P^*}
\def\Pstarsub{P^*_{\mathsf{sub}}}
\nc{\cPstar}[1]{\ensuremath{\cP^*_{#1}}}
\nc{\ctx}[3]{\ensuremath{{{\widehat{#1}}_{#2}^{(#3)}}}}

\def\Measure{\mathsf{Measure}}
\nc{\PiMeasure}{\ensuremath{\Pi_\Measure}}
\nc{\VMeasure}{\ensuremath{V_\Measure}}
\nc{\PMeasure}{\ensuremath{P_\Measure}}
\nc{\PMeasureStar}{\ensuremath{P_\Measure^*}}
\nc{\cVMeasure}[1]{\ensuremath{\cV_{\Measure,#1}}}
\nc{\cPMeasure}[1]{\ensuremath{\cP_{\Measure,#1}}}

\def\Naive{\mathsf{int}}
\nc{\PiNaive}{\ensuremath{\Pi_\Naive}}
\nc{\VNaive}{\ensuremath{V_\Naive}}
\nc{\PNaive}{\ensuremath{P_\Naive}}
\nc{\PNaiveStar}{\ensuremath{P_\Naive^*}}
\nc{\cVNaive}[1]{\ensuremath{\cV_{\Naive,#1}}}
\nc{\cPNaive}[1]{\ensuremath{\cP_{\Naive,#1}}}
\nc{\cPNaiveStar}[1]{\ensuremath{\cP_{\Naive,#1}^*}}

\nc{\Picomp}[1]{\ensuremath{\Pi_{\mathsf{comp,#1}}}}

\nc{\stepref}[1]{Step~\ref{step:#1}}

\newcommand{\proj}[1]{|#1\rangle\langle #1|}
\nc{\vev}[1]{\langle#1\rangle}
\nc{\grad}{{\vec{\nabla}}}
\nc{\abs}[1]{\lvert#1\rvert}

\DeclareMathOperator{\Hyb}{Hyb}
\DeclareMathOperator{\id}{id}

\DeclareMathOperator{\negl}{negl}

\DeclareMathOperator{\tr}{tr}

\DeclareMathOperator{\Sym}{Sym}

\DeclareMathOperator{\BPP}{\mathsf{BPP}}
\DeclareMathOperator{\QPIP}{\mathsf{QPIP}}
\DeclareMathOperator{\SampBQP}{\mathsf{SampBQP}}

\DeclareMathOperator{\BQP}{\mathsf{BQP}}

\newcommand{\be}{\begin{equation}}
\newcommand{\ee}{\end{equation}}
\newcommand{\bea}{\begin{eqnarray}}
\newcommand{\eea}{\end{eqnarray}}
\newcommand{\nn}{\nonumber}
\newcommand{\bi}{\begin{itemize}}
\newcommand{\ei}{\end{itemize}}
\newcommand{\bn}{\begin{enumerate}}
\newcommand{\en}{\end{enumerate}}
\def\beas#1\eeas{\begin{eqnarray*}#1\end{eqnarray*}}
\def\ba#1\ea{\begin{align}#1\end{align}}

\nc{\bas}{\[\begin{aligned}}
\nc{\eas}{\end{aligned}\]}

\nc{\bpm}{\begin{pmatrix}}
\nc{\epm}{\end{pmatrix}}

\def\nn{\nonumber}

\def\L{\left} 
\def\R{\right}

\nc{\given}{\ensuremath{\;\middle|\;}}

\newtheorem{theorem}{Theorem}[section]
\newtheorem{lemma}{Lemma}[section]
\newtheorem{definition}{Definition}[section]

\newtheorem{cor}{Corollary}
\newtheorem{remark}{Remark}[theorem]

\newenvironment{thm}{\begin{theorem}}{\end{theorem}}
\newenvironment{lem}{\begin{lemma}}{\end{lemma}}

\newenvironment{dfn}{\begin{definition}}{\end{definition}}

\newenvironment{prf}{\begin{proof}}{\end{proof}}

\def\eps{\epsilon}

\def\cA{\mathcal{A}}
\def\cB{\mathcal{B}}

\def\cH{\mathcal{H}}

\def\cP{\mathcal{P}}

\def\cS{\mathcal{S}}

\def\cV{\mathcal{V}}

\nc{\Prob}[1]{\ensuremath{\Pr\left[#1\right]}}

\def\bbN{\mathbb{N}}
\def\bbR{\mathbb{R}}

\def\benum{\begin{enumerate}}
\def\eenum{\end{enumerate}}
\def\bdesc{\begin{description}}
\def\edesc{\end{description}}

\nc{\myprotoref}[1]{\hyperref[#1]{Protocol~\ref*{#1}}}

\nc{\todo}[1]{\textcolor{red}{todo: #1}}

\def\begsub#1#2\endsub{\begin{subequations}\label{eq:#1}\begin{align}#2\end{align}\end{subequations}}
\nc\qand{\qquad\text{and}\qquad}
\nc\mnb[1]{\medskip\noindent{\bf #1}}
\nc\mn{\medskip\noindent}

\nc{\nl}{\nn \\ &=}  
\nc{\nnl}{\nn \\ &}  
\nc{\fot}{\frac{1}{2}} 
\nc{\oo}[1]{\frac{1}{#1}} 
\newcommand{\ben}{\begin{enumerate}}
\newcommand{\een}{\end{enumerate}}
\nc{\mc}{\mathcal}
\nc{\beq}{\begin{equation}}
\nc{\eeq}{\end{equation}}

\nc{\onenorm}[1]{\L\| #1 \R\|_1} 

\nc{\Ra}{\Rightarrow}
\nc{\zo}{\{0,1\}}

\newcommand{\defeq}{:=}

\newcommand{\ext}{\mathsf{Ext}}

\newcommand{\Acc}{\mathsf{Acc}}

\newcommand*{\regX}{\mathbf{X}}

\newcommand*{\regZ}{\mathbf{Z}}

\newcommand*{\regC}{\mathbf{C}}

\begin{document}

\date{}
\maketitle

\begin{abstract}

In a recent breakthrough, Mahadev constructed a classical verification of quantum computation (CVQC)  protocol for a  classical client to delegate decision problems in $\BQP$ to an untrusted quantum prover under computational assumptions. In this work, we explore further the feasibility of CVQC with the more general \emph{sampling} problems in BQP and with the desirable \emph{blindness} property. We contribute affirmative solutions to both as follows. 
\begin{itemize}
\item  Motivated by the sampling nature of many quantum applications (e.g., quantum algorithms for machine learning and quantum supremacy tasks), we initiate the study of  CVQC for \emph{quantum sampling problems} (denoted by $\SampBQP$).  More precisely, in a CVQC protocol for a $\SampBQP$ problem, the prover and the verifier are given an input $x\in \zo^n$ and a quantum circuit $C$, and the goal of the classical client is to learn a sample from the output $z \leftarrow C(x)$ up to a small error, from its interaction with an untrusted prover. We demonstrate its feasibility by constructing a four-message CVQC protocol for $\SampBQP$ based on the quantum \emph{Learning With Errors} assumption.

\item
The \emph{blindness} of CVQC protocols refers to a property of the protocol where the prover learns nothing, and hence is blind, about the client's input. It is a highly desirable property that has been intensively studied for the delegation of quantum computation. 
We provide a simple yet powerful \emph{generic} compiler that transforms any CVQC protocol to a blind one while preserving its completeness and soundness errors as well as the number of rounds.  
\end{itemize}
Applying our compiler to (a parallel repetition of) Mahadev's CVQC protocol for $\BQP$ and our CVQC protocol for $\SampBQP$ yields the first \emph{constant-round} blind CVQC protocol for $\BQP$ and $\SampBQP$ respectively, with negligible and inverse polynomial soundness errors respectively, and negligible completeness errors. 

\vspace{1mm}

\end{abstract}

\newpage

\section{Introduction}

Can quantum computation, with potential computational advantages that are intractable for classical computers,
be efficiently verified by classical means?
This problem has been a major open problems in quantum complexity theory and delegation of quantum computation~\cite{web:Aaronson}. A complexity-theoretic formulation of this problem by Gottesman in 2004~\cite{web:Aaronson} asks about the possibility for an efficient classical verifier (a $\BPP$ machine) to verify the output of an
efficient quantum prover (a $\BQP$ machine).
In the absence of techniques for directly tackling this question, earlier feasibility results on this problem have been focusing on two weaker formulations.
The first type of feasibility results (e.g.,~\cite{BFK09,arXiv:ABOEM17,FK17,mf16}) considers the case where the verifier is equipped with limited quantum power.
The second type of feasibility results (e.g,~\cite{Nat:RUV13, CGJV19, Gheorghiu_2015, HPF15})
considers a $\BPP$ verifier interacting with at least two entangled, non-communicating quantum provers.

Recently, the problem is resolved by a breakthrough result of Mahadev~\cite{FOCS:Mahadev18a}, who constructed the first Classical Verification of Quantum Computation (CVQC) protocol for $\BQP$, where an efficient classical ($\BPP$) verifier can interact with an efficient quantum ($\BQP$) prover to verify any $\BQP$ language. Soundness of Mahadev's protocol  is based on a widely recognized computational assumption that the learning with errors (LWE) problem~\cite{JACM:Regev09} is hard for $\BQP$ machines.
The technique invented therein has inspired many subsequent developments of CVQC protocols with improved parameters and functionality. For example, Mahadev's protocol has a large constant soundness error. The works of ~\cite{arXiv:AlaChiHun19,arXiv:ChiaChungYam19} use parallel repetition to achieve a negligible soundness error. As another example, the work of ~\cite{FOCS:GheVid19} extends Mahadev's techniques in an involved way to obtain a CVQC protocol with an additional blindness property.  

In this work, we make two more contributions to this exciting line of research. First, we observe that the literature has mostly restricted the attention to delegation of \emph{decision} problems (i.e., $\BQP$). Motivated by the intrinsic randomness of quantum computation and the sampling nature of many quantum applications, we initiate the study of CVQC for quantum \emph{sampling} problems. Second, we further investigate the desirable \emph{blindness} property and construct the first \emph{constant-round} blind CVQC protocols. We elaborate on our contributions in \Cref{subsection:sampling} and \ref{subsection:blind}, respectively.

\subsection{CVQC for Quantum Sampling Problems} \label{subsection:sampling}

We initiate the study of CVQC for quantum sampling problem, which we believe is highly desirable and natural for delegation of quantum computation. Due to the intrinsic randomness of quantum mechanics, the output from a quantum computation is randomized and described by a distribution. 
Thus, if a classical verifier want to utilize the full power of a quantum machine, the ability to get a verifiable sample from the quantum circuit's output distribution is desirable.  On a more concrete level, quantum algorithms like Shor's algorithm~\cite{Shor} has a significant quantum sampling component, and the recent quantum supremacy proposals (e.g.,~\cite{Boson, IQP, nature-google}) are built around sampling tasks, suggesting the importance of sampling in quantum computation.

It is worth noting that the difficulty of extending the delegation of decision problem to the delegation of sampling problems is quantum-specific. This is because 
there is a simple reduction from the delegation of \emph{classical} sampling problems to decision ones:  the verifier can sample and fix the random seed of the computation, 
which makes the computation deterministic. Then, the verifier can delegate the output of the computation bit-by-bit as decision problems. However, this derandomization trick does not work in the quantum setting due to its intrinsic randomness.

\vspace{-3pt}

\paragraph{Our Contribution.}
As the first step to formalize CVQC for quantum sampling problems, we consider the complexity class $\SampBQP$ introduced by Aaronson~\cite{aaronson_2013} as a natural class to capture efficiently computable quantum sampling problems. 
$\SampBQP$ consists of sampling problems $(D_x)_{x\in\zo^*}$ that can be approximately sampled by a $\BQP$ machine with a desired inverse polynomial error (See Section~\ref{sec:samp_definition} for the formal definition). We consider CVQC for a $\SampBQP$ problem $(D_x)_{x\in\zo^*}$ where a classical $\BPP$ verifier delegates the computation of a sample $z\leftarrow D_x$ for some input $x$ to a quantum $\BQP$ prover. Completeness requires that when the prover is honest, the verifier should accept with high probability and learn a correct sample $z\leftarrow D_x$. For soundness, intuitively, the verifier should not accept and output a sample with incorrect distribution  when interacting with a malicious prover. We formalize the soundness by a strong \emph{simulation-based} definition, (\Cref{dfn:stats-secure-proto-sampbqp}),
where we require that the joint distribution $(d,z)$ of the decision bit $d \in \set{\Acc, \Rej}$ and the output $z$ (which is $\bot$ when $d = \Rej$) is $\eps$-close (in  either statistical or computational sense) to an ``ideal distribution'' $(d,z_{ideal})$, where $z_{ideal}$ is sampled from the desired distribution $D_x$ when $d = \Acc$ and set to $\bot$ when $d = \Rej$.\footnote{This simulation-based formulation is analogous to the standard composable security definition for QKD.}

As our main result, we construct a constant-round CVQC protocol for $\SampBQP$, based on the quantum LWE (QLWE) assumption that the learning-with-errors problem is hard for BQP machines. 
\begin{theorem}[informal] \label{thm:qpip0-informal}
Assuming the QLWE assumption, there exists a four-message CVQC protocol for all sampling problems in $\SampBQP$ with computational soundness and negligible completeness error.
\end{theorem}

We note that since the definition of $\SampBQP$ allows an inverse polynomial error, our CVQC protocol also implicitly allows an arbitrary small inverse polynomial error in  soundness (see Section~\ref{sec:samp_definition} for the formal definition). Achieving negligible soundness error for delegating sampling problems is an intriguing open question; see Section~\ref{subsec:discussion} for further discussions.

The construction of our CVQC protocol follows the blueprint of Mahadev's construction~\cite{FOCS:Mahadev18a}. However, there are several obstacles we need to overcome along the way. To explains the obstacles and our ideas, we first present a high-level overview of Mahadev's protocol.

\paragraph{Overview of Mahadev's Protocol.}
Following~\cite{FOCS:Mahadev18a}, we define $\QPIP_{\tau}$ as classes of interactive proof systems between an (almost) classical verifier and a quantum prover, where the classical verifier has limited quantum computational capability, formalized as possessing $\tau$-qubit quantum memory.
A formal definition is given in Appendix~\ref{sec:qpip_def}. 

At a high-level, the heart of Mahadev's protocol is a measurement protocol $\PiMeasure$ that can compile an one-round $\QPIP_1$ protocol (with special properties) to a $\QPIP_0$ protocol. Note that in a $\QPIP_1$ protocol, the verifier with one-qubit memory can only measure the prover's quantum message qubit by qubit. Informally, the measurement protocol $\PiMeasure$ allows a $\BQP$ prover to ``commit to'' a quantum state $\rho$ and a classical verifier to choose an $X$ or $Z$ measurement to apply to each qubit of $\rho$ such that the verifier can learn the resulting measurement outcome. 

Thus, if an (one-round) $\QPIP_1$ verifier only applies $X$ or $Z$ measurement to the prover's quantum message, we can use the measurement protocol $\PiMeasure$ to turn the $\QPIP_1$ protocol into a $\QPIP_0$ protocol in a natural way. One additional requirement here is that the verifier's measurement choices need to be determined at the beginning (i.e., cannot depend adaptively on the intermediate measurement outcome). 

Furthermore, in $\PiMeasure$, the verifier chooses to run a ``\emph{testing}'' round or a ``\emph{Hadamard}'' round with $1/2$ probability, respectively. Informally, the testing round is used to ``test'' the commitment of $\rho$, and the Hadamard round is used to learn the measurement outcome. (See \Cref{proto:urmila4} for further details about the measurement protocol $\PiMeasure$.) Another limitation here is that in the testing round, the verifier only ``test'' the commitment without learning any measurement result. 

In~\cite{FOCS:Mahadev18a}, Mahadev's CVQC protocol for $\BQP$ is constructed by applying her measurement protocol to the one-round $\QPIP_1$ protocol of~\cite{PhysRevA.93.022326, mf16}, which has the desired properties that the verifier only performs non-adaptive $X/Z$ measurement to the prover's quantum message. The fact that the verifier does not learn the measurement outcome in the testing round is not an issue here since the verifier can simply accept when the test is passed (at the cost of a constant soundness error).

\paragraph{Overview of Our Construction.}
Following the blueprint of Mahadev's construction, our construction proceeds in the following two steps: 1. construct a $\QPIP_1$ protocol for $\SampBQP$ with required special property, and 2. compile the $\QPIP_1$ protocol using $\PiMeasure$ to get the desired $\QPIP_0$ protocol. The first step can be done by combining existing techniques from different contexts, whereas the second step is the main technical challenge.
At a high-level, the reason is the above-mentioned issue that the verifier does not learn the measurement outcome in the testing round. While this is not a problem for decision problems,
for sampling problems, the verifier needs to produce an output sample when accepts, but there seems to be no way to produce the output for the verifier without learning the measurement outcome.   
We discuss both steps in turn as follows.

\emph{$\diamond$ Construct a $\QPIP_1$ protocol for $\SampBQP$ with required special property}: Interestingly, while the notion of delegation for quantum sampling problem is not explicitly formalized in their work, Hayashi and Morimae~\cite{hayashi2015verifiable} constructed an one-round $\QPIP_1$ protocol that can delegate quantum sampling problem and achieve our notion of completeness and soundness\footnote{They did not prove our notion of soundness for their construction, but it is not hard to prove its soundness based on their analysis.}. Furthermore, their protocol has information-theoretic security and additionally achieve the blindness property. However, in their protocol, the computation is performed by the verifier using measurement-based quantum computation (MBQC)\footnote{In more detail, the prover of their protocol is required to send multiple copies of the graph states to the verifier (qubit by qubit). The verifier tests the received supposedly graph states using cut-and-choose and perform the computation using MBQC.}, and hence the verifier needs to perform adaptive measurement choices. Therefore, we cannot rely on their $\QPIP_1$ protocol for $\SampBQP$. 

Instead, we construct the desired $\QPIP_1$ protocol for $\SampBQP$ by generalizing the approach of local Hamiltonian reduction used in~\cite{PhysRevA.93.022326, mf16} to verify $\SampBQP$. Doing so requires the combination of several existing techniques from different context with some new ideas. For example, to handle $\SampBQP$, we need to prove lower bound on the spectral gap of the reduced local Hamiltonian instance, which is reminiscent to the simulation of quantum circuits by adiabatic quantum computation~\cite{adiabatic}. To achieve soundness, we use cut-and-choose and analyze it using de Finetti theorem in a  way similar to~\cite{takeuchi2018verification,hayashi2015verifiable}. See Section~\ref{sec:qpip1} for detailed discussions.

\emph{$\diamond$ Compile the $\QPIP_1$ protocol using $\PiMeasure$}: We now discuss how to use Mahadev's measurement protocol to compile the above $\QPIP_1$ protocol for $\SampBQP$ to a $\QPIP_0$ protocol. As mentioned, a major issue we need to address in Mahadev's original construction is that when the verifier $V$ chooses to run a testing round, $V$ does not learn an output sample when it accepts.  

Specifically, let $\PiNaive$ be an ``intermediate'' $\QPIP_0$ protocol obtained by applying Mahadev's compilation to the above $\QPIP_1$ protocol. In such a protocol, when the verifier $V$ chooses to run the Hadamard round, it could learn a measurement outcome from the measurement protocol and be able to run the $\QPIP_1$ verifier to generate a decision and an output sample when accepts. However, when  $V$ chooses to run the testing round, it only decides to accept/reject without being able to output a sample. 

 A natural idea to fix the issue is to execute multiple copies of $\PiNaive$ in parallel\footnote{It is also reasonable to consider sequential repetition, but we consider parallel repetition for its advantage of preserving the round complexity.}, and to choose a random copy to run the Hadamard round to generate an output sample and use all the remaining copies to run the testing round. The verifier accepts only when all executions accept and outputs the sample from the Hadamard round. We call this protocol $\PiSampZ$.

Clearly from the construction, the verifier now can output a sample when it decides to accept, and output a correct sample when interacting with an honest prover (completeness). The challenge is to show that $\PiSampZ$ is computationally sound. Since we are now in the computational setting, we cannot use the quantum de Finetti theorem as above which only holds in the information-theoretical setting. Furthermore, parallel repetition for computationally sound protocols are typically difficult to analyze, and known to not always work for protocols with four or more messages even in the classical setting~\cite{BIN97,PW12}.

 Parallel repetition of Mahadev's protocol for $\BQP$ has been analyzed before in ~\cite{arXiv:ChiaChungYam19, arXiv:AlaChiHun19}. However, the situation here is different. 
 For $\BQP$, the verifier simply chooses to run the Hadamard and testing rounds independently for each repetition.
 In contrast, our $\PiSampZ$ runs the Hadamard round in one repetition and runs the testing rounds in the rest. The reason is that in $\SampBQP$, as well as generically in sampling problems, there is no known approach to combine multiple samples to generate one sample with reduced error, i.e., there is no generic error reduction method for the sampling problem. 
In contrast, the error reduction for decision problems can be done with the majority vote. 
As a result, while the soundness error decreases exponentially for $\BQP$, as we see below (and also in the above $\QPIP_1$ protocols), for $\SampBQP$, $m$-fold repetition only decreases the error to $\poly(1/m)$. 

To analyze the soundness of $\PiSampZ$, we use the \emph{partition lemma} developed in~\cite{arXiv:ChiaChungYam19} to analyze the prover's behavior while executing copies of $\PiMeasure$.\footnote{The analysis of \cite{arXiv:AlaChiHun19} is more tailored to the decision problems setting, and it is unclear how to extend it to sampling problems where there are multiple bits of output.} Intuitively, the partition lemma says that for any cheating prover and for each copy $i\in[m]$, there exist two efficient ``projectors" \footnote{Actually they are not projectors, but for the simplicity of this discussion let's assume they are.} $G_{0,i}$ and $G_{1,i}$ in the prover's internal space with $G_{0,i}+G_{1,i} \approx Id$. $G_{0,i}$ and $G_{1,i}$ splits up the prover's residual internal state after sending back his first message.
$G_{0,i}$ intuitively represents the subspace where the prover does not knows the answer to the testing round on the $i$-th copy, while $G_{1,i}$ represents the subspace where the prover does. Note that the prover is using a single internal space for all copies, and every $G_{0,i}$ and every $G_{1,i}$ is acting on this single internal space. 
By using this partition lemma iteratively, we can decompose the prover's internal state $\ket{\psi}$ into sum of subnormalized states.
First we apply it to the first copy, writing $\ket{\psi}=G_{0,1}\ket{\psi}+G_{1,1}\ket{\psi} \equiv \ket{\psi_0}+\ket{\psi_1}$.
The component $\ket{\psi_0}$ would then get rejected as long as the first copy is chosen as a testing round,
which occurs with pretty high probability.
More precisely, the output corresponding to $\ket{\psi_0}$ is $1/m$-close to the ideal distribution that just rejects all the time.
On the other hand, $\ket{\psi_1}$ is now binding on the first copy;
we now similarly apply the partition lemma of the second copy to $\ket{\psi_1}$.
We write $\ket{\psi_1}=G_{0,2}\ket{\psi_1}+G_{1,2}\ket{\psi_1}\equiv \ket{\psi_{10}}+\ket{\psi_{11}}$, and apply the same argument about $\ket{\psi_{10}}$ and $\ket{\psi_{11}}$.
We then continue to decompose $\ket{\psi_{11}}=\ket{\psi_{110}}+\ket{\psi_{111}}$ and so on, until we reach the last copy and obtain $\ket{\psi_{1^m}}$.
Intuitively, all the $\ket{\psi_{1\dots10}}$ terms will be rejected with high probability, while the $\ket{\psi_{1^m}}$ term represents the ``good" component where the prover knows the answer to every testing round and therefore has high accept probability. Therefore, $\ket{\psi_{1^m}}$ also satisfies some binding property,
so the verifier should obtain a measurement result of some state on the Hadamard round copy,
and the soundness of the $\QPIP_1$ protocol $\PiSamp$ follows.

However, the intuition that $\ket{\psi_{1^m}}$ is binding to every Hadamard round is incorrect. As $G_{1,i}$ does not commute with $G_{1,j}$, $\ket{\psi_{1^m}}$ is unfortunately only binding for the $m$-th copy.
To solve this problem, we start with a pointwise argument and fix the Hadamard round on the $i$-th copy where $\ket{\psi_{1^i}}$ is binding,
and show that the corresponding output is $O(\norm{\ket{\psi_{1^{i-1}0}}})$-close to ideal.
We can later average out this error over the different choices of $i$, since not all $\norm{\ket{\psi_{1^{i-1}0}}}$ can be large at the same time. Another way to see this issue is to notice that we are partitioning a quantum state, not probability events, so there are some inconsistencies between our intuition and calculation. Indeed, the error we get in the end is $O(\sqrt{1/m})$ instead of the $O(1/m)$ we expected. 

The intuitive analysis outlined above glosses over many technical details, and we substantiate this outline with full details in 
Section~\ref{sec:qpip0_all}.

\subsection{Blind CVQC Protocols} \label{subsection:blind}

Another desirable property of CVQC protocols is  \emph{blindness}, which means that the prover does not learn any information about the private input for the delegated computation.\footnote{In literature, the definition of blindness may also require to additionally hide the computation. We note the two notions are equivalent from a feasibility point of view by a standard transformation (see Remark~\ref{rmk:blind-comp} in Section~\ref{sec:qpip_def}). 
} In the relaxed setting where the verifier has a limited quantum capability, Hayashi and Morimae~\cite{hayashi2015verifiable} constructed a blind $\QPIP_1$ protocol for delegating quantum computation with information-theoretic security that also handles sampling problems. 
However, for purely classical verifiers, blind CVQC protocols seem much more difficult to construct. This goal is recently achieved by the seminal work of Gheorghiu and Vidick~\cite{FOCS:GheVid19}, who constructed the first blind CVQC protocol for $\BQP$ by constructing a composable remote state preparation protocol and combining it with the verifiable blind quantum computation protocol of Fitzsimons and Kashefi~\cite{FK17}. However, their protocol has   polynomially many rounds  and requires a rather involved analysis. Before our work, it is an open question whether constant-round blind CVQC protocol for $\BQP$ is achievable.

\paragraph{Our Contribution.} Somewhat surprisingly, we provide a simple yet powerful  
\emph{generic}  compiler that transforms any CVQC protocol to a blind one while preserving completeness, soundness, as well as its round complexity. 
Our compiler relies on quantum fully homomorphic encryption (QFHE) schemes with certain ``classical-friendly'' properties, which is satisfied by both constructions of Mahadev~\cite{mahadev_qfhe} and Brakerski~\cite{brakerski_qfhe}.

\begin{theorem}[informal]
Assuming the QLWE assumption\footnote{By using Brakerski's QFHE, we only need to rely on the QLWE assumption with polynomial modulus in this theorem.}, there exists a protocol compiler that transforms any CVQC protocol $\Pi$ to a CVQC protocol $\Piblind$ that achieves blindness while preserves its round complexity, completeness, and soundness.
\end{theorem}

Applying our blindness compiler to the parallel repetition of Mahadev's protocol from~\cite{arXiv:ChiaChungYam19, arXiv:AlaChiHun19}, we obtain the first constant-round blind CVQC protocol for $\BQP$ with negligible completeness and soundness error, resolving the aforementioned open question.

\begin{theorem}[informal]
    Under the QLWE assumption, there exists a blind, four-message CVQC protocol for all languages in $\BQP$ with negligible completeness and soundness errors.
\end{theorem}

We can also apply our compier to our CVQC protocol for $\SampBQP$ to additionally achieve blindness. 

\begin{theorem}[informal]
        Under the QLWE assumption, there exists a blind, four-message CVQC protocol for all sampling problems in $\SampBQP$ with  computational soundness and negligible completeness error.
\end{theorem}

\paragraph{Techniques.} At a high-level, the idea is simple: we run the original protocol under a QFHE with the QFHE key generated by the verifier. Intuitively, this allows the prover to compute his next message under encryption without learning verifier's message, and hence achieves blindness while preserving the properties of the original protocol.
One subtlety with this approach is the fact that the verifier is classical while the QFHE cipher text could contain quantum data.
In order to make the classical verifier work in this construction, the ciphertext and the encryption/decryption algorithm need to be classical when the underlying message is classical. Fortunately, such  ``classical-friendly'' property is satisfied by the construction of~\cite{mahadev_qfhe, brakerski_qfhe}.

A more subtle issue is to preserve the soundness.
In particular, compiled protocols with a single application of QFHE might (1) leak information about the circuit evaluated by the verifier through its outputted QFHE ciphertexts (i.e., no \emph{circuit privacy});
or (2) fail to simulate original protocols upon receiving invalid ciphertexts from the prover.
We address these issues by letting the verifier switch to a fresh new key for each round of the protocol. 
Details are given in \Cref{sec:BlindBQP2}.

\vspace{-3pt}

\subsection{Related and Followup Works and Discussions}  \label{subsec:discussion}

As mentioned, while we are the first to explicitly investigate delegation of quantum sampling problems, Hayashi and Morimae~\cite{hayashi2015verifiable} constructed an one-round blind $\QPIP_1$ protocol that can be used to delegate $\SampBQP$ and achieve our notion of information-theoretical security. Like our $\SampBQP$ protocol, their protocol has an arbitrarily small inverse polynomial soundness error instead of negligible soundness error. Also as mentioned,  Gheorghiu and Vidick~\cite{FOCS:GheVid19} constructed the first blind CVQC protocol for $\BQP$ by constructing a composable remote state preparation protocol and combining it with the verifiable blind quantum computation protocol of Fitzsimons and Kashefi~\cite{FK17}. However, their protocol has polynomially many rounds  and requires a rather involved analysis. 

It is also worth noting that several existing constructions in the relaxed models (e.g., verifiable blind computation~\cite{FK17}) can be generalized to delegate $\SampBQP$ in a natural way, but it seems challenging to analyze the soundness of the generalized protocol. Furthermore, it is unlikely that these generalized protocols can achieve negligible soundness error for $\SampBQP$. The reason is that in all these constructions, some form of cut and choose are used to achieve soundness.
For sampling problems, as mentioned, there seems to be no generic way to combine multiple samples for error reduction, so the verifier needs to choose one sample to output in the cut and choose. In this case, an adversarial prover may choose to cheat on a random copy in the cut and choose and succeed in cheating with an inverse polynomial probability. 

On the other hand, while the definition of $\SampBQP$ in~\cite{aaronson_2013, Boson} allows an inverse polynomial error, there seems to be no fundamental barriers to achieve negligible error. It is conceivable that negligible error can be achieved using quantum error correction. Negligible security error is also achievable in the related settings of secure multi-party quantum computation ~\cite{CGS02,DNS12} and verifiable quantum FHE~\cite{ADSS17} based on verifiable quantum secret sharing or quantum authentication codes\footnote{The security definitions are not comparable, but it seems plausible that the techniques can be used to achieve negligible soundness error for sampling problems.}. However, both primitives require computing and communicating quantum encodings and are not applicable in the context of CVQC and $\QPIP_1$. An intriguing open problem is whether it is possible to achieve negligible soundness error with classical communication while delegating a quantum sampling problem.

In a recent work, Bartusek~\cite{bartuseksecure} used the technique we developed for delegation of $\SampBQP$ to construct secure quantum computation protocols with classical communication for pseudo-deterministic quantum functionalities.

\paragraph{Organization}
We provide preliminaries on technical background in Section~\ref{sec:prelim}. 
Our simulation-based definition of CVQC for $\SampBQP$ is discussed in Section~\ref{sec:samp_definition}. 
Our main technical contributions are explained in Section~\ref{sec:sampbqp_short} (a construction of $\QPIP_1$ protocol for $\SampBQP$), 
Section~\ref{sec:qpip0_all} (the construction of $\QPIP_0$ protocol for $\SampBQP$ based on the above $\QPIP_1$ protocol), 
and Section~\ref{sec:BlindBQP2} (a generic compiler to upgrade $\QPIP_0$ protocols with blindness).

\section{Preliminaries} \label{sec:prelim}

\subsection{Notations}

Let $\mathcal{B}$ be the Hilbert space corresponding to one qubit. Let $H:\mathcal{B}^{\otimes n}\rightarrow\mathcal{B}^{\otimes n}$ be Hermitian matrices. We use $H\geq0$ to denote $H$ being positive semidefinite. Let $\lambda(H)$ be the smallest eigenvalue of $H$. The ground states of $H$ are the eigenvectors corresponding to $\lambda(H)$. For matrix $H$ and subspace $S$, let $H\big|_S=\Pi_S H \Pi_S$, where $\Pi_S$ is the projector onto the subspace $S$.  
We write $x\xleftarrow{\$}X$ when $x$ is sampled uniformly at random from the set $X$.

The phase gate and Pauli matrices are the following matrices.

\begin{definition}
    $P(i)=\begin{pmatrix}1&0\\0&i\end{pmatrix}$,
    $X=\begin{pmatrix}0&1\\1&0\end{pmatrix}$,
    $Y=\begin{pmatrix}0&-i\\i&0\end{pmatrix}$,
    $Z=\begin{pmatrix}1&0\\0&-1\end{pmatrix}$.
\end{definition}

For $n$-qubit states $\rho_1, \rho_2\in\cB^{\otimes n}$,
let $\norm{\rho_1-\rho_2}_{\mathrm{tr}}=\frac{1}{2}\norm{\rho_1-\rho_2}_1$ denote the trace distance between $\rho_1$ and $\rho_2$. We say $\rho_1$ and $\rho_2$ are $\eps$-close if $\norm{\rho_1- \rho_2}_{\mathrm{tr}}\leq\eps$.
For two distributions $f_1$ and $f_2$ over the same finite domain $X$, let $\norm{f_1-f_2}_{\mathrm{TV}}=\frac{1}{2}\sum_{x\in X}\abs{f_1(x)-f_2(x)}$ denote their total variation distance,
and we say $f_1$ and $f_2$ are $\eps$-close if $\norm{f_1-f_2}_{\mathrm{TV}}\leq\eps$.
We additionally denote the following for states that are only computationally indistinguishable (but might be statistically far):
\begin{definition}
\label{def:comp-indis}
    Two ensemble of states $\set{\rho_{1,\lambda}}_{\lambda\in\bbN}$ and $\set{\rho_{2,\lambda}}_{\lambda\in\bbN}$ are \emph{$\eps$-computationally indistinguishable} if for all $\BQP$ distinguishers $\mathsf{D}$ and $\lambda\in\bbN$,
    $$\abs{\Pr\left[\mathsf{D}(\rho_{1,\lambda})=1\right]-\Pr\left[\mathsf{D}(\rho_{2,\lambda})=1\right]}\leq\eps(\lambda).$$
    We drop the $\eps$ parameter when it is negligible.
    We extend computational indistinguishability to subnormalized states by interpreting a subnormalized state $\rho$ as sending out $\bot$ with probability $1-\tr(\rho)$ and having the distinguisher $\mathsf{D}$ output $\bot$ whenever he receives $\bot$.
    We also extend it to classical random variables by treating them as mixed states.
\end{definition}

We denote measurement as follows:
\begin{definition} [quantum-classical channels]
    \label{def:QCChannel}
    A quantum measurement is given by a set of matrices $\set{M_k}$ such that $M_k\geq0$ and $\sum_k M_k=\id$.
    We associate to any measurement a map $\Lambda(\rho)=\sum_k \tr(M_k\rho)\proj{k}$
    with $\{\ket{k}\}$ an orthonormal basis.
    This map is also called a \emph{quantum-classical channel}.
\end{definition}
Furthermore, we associate measurement outcomes in $X$ or $Z$ basis with a corresponding random variable as follows:
\begin{definition}[$M_{XZ}(\rho,h)$]
	For any natural number $n$, $n$-bit string $h$, and $n$-qubit quantum state $\rho$, consider the following measurement procedure: measure the first qubit of $\rho$ in $X$ basis if $h_1=0$; measure the first qubit of $\rho$ in $Z$ basis if $h_1=1$.  Measure the second qubit of $\rho$ in $X$ basis if $h_2=0$; measure the second qubit of $\rho$ in $Z$ basis if $h_2=1$. Continue qubit-by-qubit until all $n$-qubits of $\rho$ are measured, where $i$-th qubit is measured in $X$ basis if $h_i=0$ and  $i$-th qubit is measured in the $Z$ basis if $h_i=1$.

	We denote the $n$-bit random variable corresponding to the measurement results as $M_{XZ}(\rho,h)$.
\end{definition}

\subsection{Quantum Prover Interactive Proof (QPIP)} 
\label{sec:qpip_def}

Following \cite{FOCS:Mahadev18a}, we model the interaction between an almost classical client with quantum memory of size $\tau$ (verifier) and a quantum server (prover) in the following definition.
When the client has limited quantum memories (i.e., $\tau > 0$), it becomes a ``streaming" model. Namely, the server can send qubits one-by-one to let the client measure them sequentially, even if the client cannot hold the entire quantum message.
When the client is fully classical (i.e., $\tau = 0$), this model specializes to the standard interactive protocols.

\begin{definition}
    A protocol $\Pi=(P, V)$ is said to be in $\QPIP_\tau$ if it satisfies the following properties:
    \begin{itemize}
        \item $P$ is a $\BQP$ machine, which also has access to a quantum channel that can transmit $\tau$ qubits to the verifier per use.
        \item $V$ is a hybrid machine of a classical part and a limited quantum part. The classical part is a $\BPP$ machine. The quantum part is a register of $\tau$ qubits, on which the verifier can perform arbitrary quantum operations and which has access to a quantum channel that can transmit $\tau$ qubits. At any given time, the verifier is not allowed to possess more than $\tau$ qubits. The interaction between the quantum and classical parts of the verifier is the usual one: the classical part controls which operations are to be performed on the quantum register, and outcomes of measurements of the quantum register can be used as input to the classical part.
        \item There is also a classical communication channel between the prover and the verifier, which can transmit polynomially many bits to either direction.
    \end{itemize}
\end{definition}

It is straightforward to define what it means for a $\QPIP_\tau$ protocol to decide a $\BQP$ language.

\begin{definition}
    Let $\Pi=(P, V)$ be a $\QPIP_\tau$ protocol.
    We say it is a protocol for a $\BQP$ language $L$ with completeness error $c(\cdot)$ and soundness error $s(\cdot)$ if the following holds:
\begin{itemize}
        \item On public input $1^\lambda$ and $x\in\zo^*$, the verifier outputs either $\Acc$ or $\Rej$.
        \item (Completeness): For all security parameters $\lambda\in\bbN$ and $x\in\zo^{\poly(\lambda)}$, if $x\in L$ then
        $\Pr\left[(P, V)(x, 1^\lambda)=\Rej\right]<c(\lambda)$.
        \item (Soundness): For all cheating $\BQP$ provers $P^*$, sufficiently large \\security parameters $\lambda\in\bbN$, and $x\in\zo^{\poly(\lambda)}$, if $x \notin L$ then \\$\Pr\left[(P, V)(x, 1^\lambda)=\Acc\right]<s(\lambda)$.
    \end{itemize}
\end{definition}

We are particularly interested in the case that $\tau = 0$, i.e., when the verifier $V$ is classical. In this case, we say $\Pi$ is a CVQC protocol for the $\BQP$ language $L$.

We now define \emph{blindness} which means that the prover does not learn any information about the verifier's input except for its length.
As such, in this context, we consider the input $x$ as the verifier's private input instead of a common input.  

\begin{definition}[Blindness]
    Let $\Pi=(P, V)$ be an interactive protocol with common inputs $1^\lambda$ 
    and verifier's private input $x\in\zo^*$.
    We say $\Pi$ is \emph{blind} if for all cheating $\BQP$ provers $P^*$ the following ensembles are computationally indistinguishable over $\lambda$:
    \begin{itemize}
        \item $\set{\View_{P^*}(P^*, V(x))(1^\lambda)}_{\lambda\in\bbN, \eps\in(0,1), x\in\zo^*}$,
        \item $\set{\View_{P^*}(P^*, V(0))(1^\lambda)}_{\lambda\in\bbN, \eps\in(0,1), x\in\zo^*}$.
    \end{itemize}
\end{definition}

\begin{remark} \label{rmk:blind-comp}
In literature, the definition of blindness may also require to additionally hide the computation. The two notions are equivalent from a feasibility point of view by a standard transformation (e.g., as in~\cite{GHRW14} in a different context) that takes the description of the computation as part of the input and delegate the computation of  universal quantum circuits.
\end{remark}

\subsection{Quantum Homomorphic Encryption Schemes}

We use the quantum fully homomorphic encryption scheme given in \cite{mahadev_qfhe} which is compatible with our use of a classical client.
\begin{definition}
    A quantum leveled homomorphic (public-key) encryption scheme $\QHE=(\QGen, \QEnc, \QEval, \QDec)$ is quadruple of quantum polynomial-time algorithms which operate as follows:
    \begin{itemize}
        \item \emph{Key generation.}
            The algorithm $(pk, sk)\leftarrow\QGen(1^\lambda, 1^L)$ takes a unary representation of the security parameter and a unary representation of the level parameter as inputs and outputs a public key $pk$ and a secret key $sk$.
        \item \emph{Encryption.}
            The algorithm $c\leftarrow\QEnc(pk, \mu)$ takes the public key $pk$ and a single bit message $\mu\in\zo$ as inputs and outputs a ciphertext $c$.
        \item \emph{Decryption.}
            The algorithm $\mu^*\leftarrow\QDec(sk, c)$ takes the secret key $sk$ and a ciphertext $c$ as inputs and outputs a message $\mu^*\in\zo$.
            It must be the case that
                $$\QDec(sk, \QEnc(pk, \mu))=\mu$$
            with overwhelming probability in $\lambda$.
        \item \emph{Homomorphic Evaluation.}
            The algorithm $c_f\leftarrow\QEval(pk, f, c_1, \ldots, c_l)$ takes the public key $pk$, a quantum circuit $f$ of depth at most $L$, and a set of $l$ ciphertexts $c_1,\cdots,c_l$ as inputs and outputs a ciphertext $c_f$.
            It must be the case that
            \begin{align*}
                \QDec&(sk, \QEval(pk, f, c_1, \ldots, c_l))
                \\
                &=f(\QDec(sk, c_1),\ldots,\QDec(sk, c_l))
            \end{align*}
            with overwhelming probability in $\lambda$.
    \end{itemize}
\end{definition}

\begin{definition}[Compactness and Full Homomorphism]
    A quantum homomorphic encryption scheme $\QHE$ is \emph{compact} if
    there exists a polynomial $s$ in $\lambda$ such that the output length of $\QEval$ is at most $s$ bits long (regardless of $f$ or
    the number of inputs). A compact scheme is (pure) \emph{fully homomorphic} if it can evaluate any efficiently computable
    Boolean function.
\end{definition}

We also recall the security definition for a FHE scheme.

\begin{definition}
    A scheme $\QHE$ is IND-CPA secure if, for any polynomial time adversary $\cA$, there exists a negligible function $\mu(\cdot)$ such that
    $$\abs{Pr[\cA(pk, \mathsf{HE.Enc}_{pk}(0))=1]-Pr[\cA(pk, \mathsf{HE.Enc}_{pk}(1))=1]}\leq\mu(\lambda),$$
    where $(pk, sk)\leftarrow\QGen(1^\lambda)$.
\end{definition}

The quantum homomorphic encryption scheme $\mathsf{QHE}$ from \cite{mahadev_qfhe} has additional properties that facilitate the use of classical clients.

\begin{definition} \label{def:classical-friendly}
    we call a quantum homomorphic encryption scheme \emph{classical-friendly} if it has the following properties:
    \begin{itemize}
        \item $\QGen$ is a classical algorithm
        \item In the case where the plaintext is classical, $\QEnc$ can be done classically.
        \item When the underlying message is a classical-quantum state, then the cipher text is composed of a corresponding classical part and a corresponding quantum part. Moreover, the classical part has a classical ciphertext and can be decrypted classically.
    \end{itemize}
\end{definition}

\begin{theorem} [Theorem 1.1 in \cite{mahadev_qfhe}]
    Under the assumption that the learning with errors problem with superpolynomial noise ratio is computationally intractable for an efficient quantum machine,
    there exists a classical-friendly compact quantum leveled homomorphic encryption scheme.
\end{theorem}

\section{Delegation of Quantum Sampling Problems} \label{sec:samp_definition}

In this section, we formally introduce the task of delegation for quantum sampling problems. We start by recalling the complexity class $\SampBQP$ defined by Aaronson~\cite{aaronson_2013, Boson}, which captures the class of sampling problems that
are approximately solvable by polynomial-time quantum algorithms.

\begin{definition} [Sampling Problem]
    \label{dfn:sampling-problem}
    A \emph{sampling problem} is a collection of probability distributions $(D_x)_{x\in\set{0, 1}^*}$, one for each input string $x\in\set{0,1}^n$, where $D_x$ is a distribution over $\set{0,1}^{m(n)}$ for some fixed polynomial $m$.
\end{definition}

\begin{definition} [$\SampBQP$]
    $\SampBQP$ is the class of sampling problems $\left(D_x\right)_{x\in\set{0, 1}^*}$ that can be (approximately) sampled by polynomial-size uniform quantum circuits. Namely, there exists a Turing machine $M$ such that for every $n \in \bbN$ and $\eps \in (0,1)$, $M(1^n, 1^{1/\eps})$ outputs a quantum circuit $C$ in $\poly(n, 1/\eps)$ time such that for every $x \in \zo^n$, the output of $C(x)$ (measured in standard basis) is $\eps$-close to $D_x$ in the total variation distance. 
\end{definition}

Note that in the above definition, there is an accuracy parameter $\eps$ and the quantum sampling algorithm only requires to output a sample that is $\eps$-close to the correct distribution in time $\poly(n,1/\eps)$.
\cite{aaronson_2013, Boson} discussed multiple reasons for allowing the inverse polynomial error, such as to take into account the inherent noise in conceivable physical realizations of quantum computer.
On the other hand, it is also meaningful to require negligible error. As discussed, it is an intriguing open question to delegate quantum sampling problem with negligible error.

We next define what it means for a $\QPIP_\tau$ protocol\footnote{See Appendix~\ref{sec:qpip_def} for a formal definition of $\QPIP_\tau$.} to solve a $\SampBQP$ problem $\left(D_x\right)_{x\in\set{0, 1}^*}$.
Since sampling problems come with an accuracy parameter $\eps$, we let the prover $P$ and the verifier $V$ receive the input $x$ and $1^{1/\eps}$ as common inputs. 
Completeness is straightforward to define, which requires that when the prover $P$ is honest, the verifier $V$ should accept with high probability and output a sample $z$ distributed close to $D_x$ on input $x$. Defining soundness is more subtle. Intuitively, it requires that the verifier $V$ should never be ``cheated'' to accept and output an incorrect sample even when interacting with a malicious prover. We formalize this by a strong simulation-based definition, where we require that the joint distribution of the decision bit $d \in \set{\Acc, \Rej}$ and the output $z$ (which is $\bot$ when $d = \Rej$) is $\eps$-close (in either statistical or computational sense) to an ``ideal distribution'' $(d,z_{ideal})$, where $z_{ideal}$ is sampled from $D_x$ when $d = \Acc$ and set to $\bot$ when $d = \Rej$. Since the protocol receives the accuracy parameter $1^{1/\eps}$ as input to specify the allowed error, we do not need to introduce an additional soundness error parameter in the definition.

\begin{definition}
    \label{dfn:stats-secure-proto-sampbqp}
    Let $\Pi=(P, V)$ be a $\QPIP_\tau$ protocol.
    We say it is a protocol for the $\SampBQP$ instance $(D_x)_{x\in\zo^*}$ with completeness error $c(\cdot)$ and statistical (resp., computational) soundness if the following holds:
    \begin{itemize}
        \item On public inputs $1^\lambda$, $1^{1/\eps}$, and $x\in\zo^{\poly(\lambda)}$, $V$ outputs $(d, z)$ where $d\in\set{\Acc, \Rej}$.
            If $d=\Acc$ then $z\in\zo^{m(\abs{x})}$ where $m$ is given in \Cref{dfn:sampling-problem}, otherwise $z=\bot$.
        \item (Completeness):
            For all accuracy parameters $\eps(\lambda)=\frac{1}{\poly(\lambda)}$,
            security parameters $\lambda\in\bbN$,
            and $x\in\zo^{\poly(\lambda)}$,
            let $(d, z)\leftarrow(P, V)(1^\lambda, 1^{1/\eps}, x)$, then $d=\Rej$ with probability at most $c(\lambda)$.
        \item (Statistical soundness): For all cheating provers $P^*$,
            accuracy parameters $\eps(\lambda)=\frac{1}{\poly(\lambda)}$,
            sufficiently large $\lambda\in\bbN$, and $x\in\zo^{\poly(\lambda)}$,
            consider the following experiment:
            \begin{itemize}
                \item Let $(d, z)\leftarrow(P^*, V)(1^\lambda, 1^{1/\eps}, x)$.
                \item Define $z_{ideal}$ by
                $$\begin{cases}
                    z_{ideal}=\bot & \text{if } d=\Rej\\
                    z_{ideal}\leftarrow D_x & \text{if } d=\Acc
                \end{cases}$$.
            \end{itemize}
            It holds that $\norm{(d,z)-(d,z_{ideal})}_{\mathrm{TV}}\leq\eps$.
		\item (Computational soundness):
        For all cheating $\BQP$ provers $P^*$, $\BQP$ distinguishers $\mathsf{D}$, accuracy parameters $\eps(\lambda)=\frac{1}{\poly(\lambda)}$,
            sufficiently large $\lambda\in\bbN$, and all $x\in\zo^{\poly(\lambda)}$,
            let us define $d, z, z_{ideal}$ by the same experiment as above.
            It holds that $(d, z)$ is $\eps$-computationally indistinguishable to $(d, z_{ideal})$ over $\lambda$.
    \end{itemize}
\end{definition}

As in the case of $\BQP$, we are particularly interested in the case that $\tau = 0$, i.e., when the verifier $V$ is classical. In this case, we say that $\Pi$ is a CVQC protocol for the $\SampBQP$ problem $(D_x)_{x\in\zo^*}$.

\section{Construction of the $\QPIP_1$ Protocol for $\SampBQP$}
\label{sec:sampbqp_short}

As we mentioned in this introduction, we will employ the circuit \emph{history} state in the original construction of the Local Hamiltonian problem~\cite{kitaev2002classical} to encode the circuit information for $\SampBQP$.
However, there are distinct requirements between certifying the computation for $\BQP$ and $\SampBQP$ based on the history state.
For any quantum circuit $C$ on input $x$, the original construction for certifying $\BQP$\footnote{The original construction is for the purpose of certifying problems in QMA. We consider its simple restriction to problems inside BQP.} consists of local Hamiltonian $\Hin, \Hclock, \Hprop$, $\Hout$ where $\Hin$ is used to certify the initial input $x$, $\Hclock$ to certify the validness of the clock register,  $\Hprop$ to certify the gate-by-gate evolution according to the circuit description, and $\Hout$ to certify the final output.
In particular, the corresponding history state is in the ground space of $\Hin$, $\Hclock$, and $\Hprop$. Note that $\BQP$ is a decision problem and its outcome (0/1) can be easily encoded into the energy $\Hout$ on the single output qubit.
As a result, the outcome of $\BQP$ can simply be encoded by the \emph{ground energy} of $\Hin + \Hclock+\Hprop + \Hout$.

To deal with $\SampBQP$, we will still employ $\Hin, \Hclock$, and $\Hprop$ to certify the circuit's input, the clock register, and gate-by-gate evolution. However, in $\SampBQP$, we care about the entire final state of the circuit, rather than the energy on the output qubit.
Our approach to certify the entire final state (which is encoded inside the history state) is to make sure that the history state is the unique ground state of $\Hin + \Hclock+ \Hprop$ and all other orthogonal states will have much higher energies.
Namely, we need to construct some $\Hin'+ \Hclock'+ \Hprop'$ with the history state as the unique ground state and with a large \emph{spectral} gap between the ground energy and excited energies.
It is hence guaranteed that any state with close-to-ground energy must also be close to the history state.
We remark that this is a different requirement from most local Hamiltonian constructions that focus on the ground energy.
We achieve so by using the \emph{perturbation} technique developed in~\cite{kempe_kitaev_regev_2006} for reducing the locality of Hamiltonian.
Another example of local Hamiltonian construction with a focus on the spectral gap can be found in~\cite{adiabatic}, where the purpose is to simulate quantum circuits by adiabatic quantum computation.

We need two more twists for our purpose.
First, as we will eventually measure the final state in order to obtain classical samples, we need that the final state occupies a large fraction of the history state. We can simply add dummy identity gates.
Second, as we are only able to perform $X$ or $Z$ measurement by techniques from~\cite{FOCS:Mahadev18a},
we need to construct X-Z only local Hamiltonians.
Indeed, this has been shown possible in, e.g.,~\cite{PhysRevA.78.012352}, which serves as the starting point of our construction.

We present the formal construction of our $\QPIP_1$ protocol $\PiSamp$ for $\SampBQP$ in \Cref{sec:sampbqp}, \Cref{ProtoQPIP1}. The soundness and completeness of \Cref{ProtoQPIP1} is stated in the following theorem, whose proof is also deferred to \Cref{sec:sampbqp}.

\begin{thm}
    \label{QPIP1thm}
	$\PiSamp$ is a $\QPIP_1$ protocol for the $\SampBQP$ problem  $(D_x)_{x\in\set{0,1}^*}$ with negligible completeness error and is statistically sound\footnote{The soundness and completeness of a $\SampBQP$ protocol is defined in \Cref{dfn:stats-secure-proto-sampbqp}.} where the verifier only needs to do non-adaptive $X/Z$ measurements.
\end{thm}

\vspace{-3pt}

\section{$\SampBQP$ Delegation Protocol for Fully Classical Client}
\label{sec:qpip0_all}

In this section, we create a delegation protocol for $\SampBQP$ with fully classical clients by adapting the approach taken in \cite{FOCS:Mahadev18a}. In \cite{FOCS:Mahadev18a}, 
Mahadev designed a protocol $\PiMeasure$ (\Cref{proto:urmila4}) that allows a $\BQP$ prover to ``commit a state" for a classical verifier to choose a $X$ or $Z$ measurement and obtain corresponding measurement results.
Composing it with the $\QPIP_1$ protocol for $\BQP$ from \cite{mf16} results in a $\QPIP_0$ protocol for $\BQP$.
In this work, we will compose $\PiMeasure$ with our $\QPIP_1$ protocol $\PiSamp$ (\Cref{ProtoQPIP1}) for $\SampBQP$ in order to obtain a $\QPIP_0$ protocol for $\SampBQP$. 

A direct composition of $\PiSamp$ and $\PiMeasure$, however, results in $\PiNaive$ (\Cref{proto:qpip0_naive}) which does not provide reasonable completeness or accuracy guarantees.
As we will see, this is due to $\PiMeasure$ itself having peculiar and weak guarantees:
the client doesn't always obtain measurement outcomes even if the server were honest.
When that happens under the $\BQP$ context, the verifier can simply accept the prover at the cost of some soundness error;
under our $\SampBQP$ context, however, we must run many copies of $\PiNaive$ in parallel so the verifier can generate its outputs from some copy.
We will spend the majority of this section analyzing the soundness of this parallel repetition.

\subsection{Mahadev's measurement protocol}\label{sec:urmila4}

$\PiMeasure$ is a 4-round protocol between a verifier (which corresponds to our client) and a prover (which corresponds to our server).
The verifier (secretly) chooses a string $h$ specifying the measurements he wants to make, and generates keys $pk, sk$ from $h$. It sends $pk$ to the prover. The prover ``commits" to a state $\rho$ of its choice using $pk$ and replies with its commitment $y$.
The verifier must then choose between two options: do a \emph{testing round} or a \emph{Hadamard round}.
In a testing round the verifier can catch cheating provers,
and in a Hadamard round the verifier receives some measurement outcome.
He sends his choice to the prover, and the prover replies accordingly. If the verifier chose testing round, he checks the prover's reply against the previous commitment, and rejects if he sees an inconsistency. If the verifier chose Hadamard round, he calculates $M_{XZ}(\rho, h)$ based on the reply.
We now formally describe the interface of $\PiMeasure$ while omitting the implementation details.

\begin{protocol}{Mahadev's measurement protocol $\PiMeasure=(\PMeasure, \VMeasure)$}
	\label{proto:urmila4}

	Inputs:
	\begin{itemize}
		\item Common input: Security parameter $1^\lambda$ where $\lambda\in\bbN$.
		\item Prover's input: a state $\rho\in\cB^{\otimes n}$ for the verifier to measure.
		\item Verifier's input: the measurement basis choice $h \in \zo^n$
	\end{itemize}

	Protocol:
	\begin{enumerate}
		\item \label{step:measure1} The verifier generates a public and secret key pair $(pk, sk)\leftarrow\cVMeasure{1}($ $1^\lambda, h)$. It sends $pk$ to the prover.
		\item \label{step:measure2} The prover generates $(y, \sigma)\leftarrow\cPMeasure{2}(pk, \rho)$.
			$y$ is a classical ``commitment", and $\sigma$ is some internal state.
			He sends $y$ to the verifier.
		\item \label{step:measure3} The verifier samples $c\xleftarrow{\$}\zo$ uniformly at random and sends it to the prover. $c=0$ indicates a \emph{testing round}, while $c=1$ indicates a \emph{Hadamard round}.
		\item \label{step:measure4} The prover generates a classical string $a\leftarrow\cPMeasure{4}(pk, c, \sigma)$ and sends it back to the verifier.
		\item \label{step:output} If it is a testing round ($c=0$), then the verifier generates and outputs $o\leftarrow\cVMeasure{T}(pk, y, a)$ where $o\in\set{\Acc, \Rej}$.
			If it is a Hadamard round ($c=1$), then the verifier generates and outputs $v\leftarrow\cVMeasure{H}($ $sk, h, y, a)$.
	\end{enumerate}
\end{protocol}

$\PiMeasure$ has negligible completeness errors, i.e. if both the prover and verifier are honest, the verifier accepts with overwhelming probability and his output on Hadamard round is computationally indistinguishable from $M_{XZ}(\rho, h)$. As for soundness,
it gives the following \emph{binding property} against cheating provers:
if a prover would always succeed on the testing round, then there exists some $\rho$ so that for any $h$ the verifier obtains $M_{XZ}(\rho, h)$ if he had chosen the Hadamard round.

\begin{lemma}[binding property of $\PiMeasure$; special case of Claim 7.1 in \cite{FOCS:Mahadev18a}]
	\label{lem:urmila-binding}
	Let $\PMeasureStar$ be a $\BQP$ cheating  prover for $\PiMeasure$ and $\lambda$ be the security parameter. Let $1-p_{h,T}$ be the  probability that the verifier accepts $\PMeasureStar$ in the testing round on basis choice $h$.\footnote{Compared to Claim 7.1 of \cite{FOCS:Mahadev18a}, we don't have a $p_{h,H}$ term here. This is because on rejecting a Hadamard round, the verifier can output a uniformly random string, and that is same as the result of measuring $h$ on the totally mixed state.} Under the QLWE assumption, there exists some $\rho^*$ so that for all verifier's input $h \in \zo^n$, the verifier's outputs on the Hadamard round is $\sqrt{p_{h,T}}+\negl(n)$-computationally indistinguishable from $M_{XZ}(\rho^*, h)$.
\end{lemma}

We now combine $\PiMeasure$ with our $\QPIP_1$ Protocol for $\SampBQP$, $\PiSamp=(\PSamp, \VSamp)$ (\Cref{ProtoQPIP1}), to get a corresponding $\QPIP_0$ protocol $\PiNaive$.
Recall that in $\PiSamp$ the verifier takes $X$ and $Z$ measurements on the prover's message.
In $\PiNaive$ we let the verifier use $\PiMeasure$ to learn those measurement outcomes instead.

\begin{protocol}{Intermediate $\QPIP_0$ protocol $\PiNaive$ for the $\SampBQP$ problem $(D_x)_{x\in\set{0, 1}^*}$}
	\label{proto:qpip0_naive}

	Inputs:
	\begin{itemize}
		\item Security parameter $1^\lambda$ where $\lambda\in\bbN$
		\item Error parameter $\eps\in(0, 1)$
		\item Classical input $x\in\zo^n$ to the $\SampBQP$ instance
	\end{itemize}

	Protocol:
	\begin{enumerate}
		\item \label{step:naive1} The verifier chooses a $XZ$-measurement $h$ from the distribution specified in \stepref{qpip1-verify} of $\PiSamp$.
		\item \label{step:naive2} The prover prepares $\rho$ by running \stepref{qpip1-state-gen} of $\PiSamp$.
		\item \label{step:urmila-in-naive}
			The verifier and prover run $(\PMeasure(\rho), \VMeasure(h))(1^\lambda)$.
			\begin{enumerate}
				\item The verifier samples $(pk, sk)\leftarrow\cVNaive{1}(1^\lambda, h)$ and sends $pk$ to the prover, where $\cVNaive{1}$ is the same as $\cVMeasure{1}$ of \Cref{proto:urmila4}. 
				\item The prover runs $(y, \sigma)\leftarrow\cPNaive{2}(pk, \rho)$ and sends $y$ to the verifier, where $\cPNaive{2}$ is the same as $\cPMeasure{2}$.
					Here we allow the prover to abort by sending $y=\bot$, which does not benefit cheating provers but simplifies our analysis of parallel repetition later.
				\item\label{step:c-urmila-in-naive} The verifier samples $c\xleftarrow{\$}\zo$ and sends it to the prover.
				\item The prover replies $a\leftarrow\cPNaive{4}(pk, c, \sigma)$.
				\item
					If it is a testing round, the verifier accepts or rejects based on the outcome of $\PiMeasure$.
					If it is a Hadamard round, the verifier obtains $v$.
			\end{enumerate}
		\item \label{step:naive-output} If it's a Hadamard round, the verifier finishes the verification step of Protocol~\ref{ProtoQPIP1} by generating and outputting $(d, z)$
	\end{enumerate}
\end{protocol}

There are several problems with using $\PiNaive$ as a $\SampBQP$ protocol. First, since the verifier doesn't get a sample if he had chosen the testing round in Step~\ref{step:c-urmila-in-naive}, the protocol has completeness error at least $1/2$. Moreover, since $\PiMeasure$ does not check anything on the Hadamard round, a cheating prover can give up passing the testing round and breaks the commitment on the Hadamard round, with only a constant $1/2$ probability of being caught.
However, we can show that $\PiNaive$ has a binding property similar to $\PiMeasure$:
if a cheating prover $\PNaiveStar$ passes the testing round with overwhelming probability whenever it doesn't abort on the second message,
then the corresponding output $(d, z)\leftarrow(\PNaiveStar, \VNaive)$ is close to $(d, z_{ideal})$.
Recall the ideal output is
$$\begin{cases}
	z_{ideal}=\bot & \text{if } d=\Rej\\
	z_{ideal}\leftarrow D_x & \text{if } d=\Acc.
\end{cases}$$
This binding property is formalized in \Cref{lem:naive-qpip0-binding}.
Intuitively,  the proof of \Cref{lem:naive-qpip0-binding}  combines the binding property of \Cref{proto:qpip0_naive} (\Cref{lem:urmila-binding}) and $\PiSamp$'s soundness (\Cref{QPIP1thm}). There is a technical issue that \Cref{proto:qpip0_naive} allows the prover to abort while \Cref{proto:urmila4} does not. This issue is solved by constructing another $\BQP$ prover $\Pstar$ for every cheating prover $\PNaiveStar$. 
Specifically, $\Pstar$ uses $\PNaiveStar$'s strategy when it doesn't abort, otherwise honestly chooses the totally mixed state for the verifier to measure.

\begin{theorem}[binding property of $\PiNaive$]
	\label{lem:naive-qpip0-binding}
	Let $\PNaiveStar$ be a cheating $\BQP$ prover for $\PiNaive$ and $\lambda$ be the security parameter.
	Suppose that $\Prob{d=\Acc\mid y\ne\bot, c=0}$ is overwhelming, 
	under the QLWE assumption, then the verifier's output in the Hadamard round is $O(\eps)$-computationally indistinguishable from $(d, z_{ideal})$.
\end{theorem}
\begin{proof}[\Cref{lem:naive-qpip0-binding}]
	We first introduce the \emph{dummy strategy} for $\PiMeasure$, where the prover chooses $\rho$ as the maximally mixed state and executes the rest of the protocol honestly.
	It is straightforward to verify that this prover would be accepted in the testing round with probability $1-\negl(\lambda)$,
	but has negligible probability passing the verification  after the Hadamard round.

	Now we construct a cheating $\BQP$ prover for \Cref{proto:qpip0_naive}, $\Pstar$, that does the same thing as $\PNaiveStar$ except at Step~\ref{step:urmila-in-naive}, where the prover and verifier runs \Cref{proto:urmila4}. $\Pstar$ does the following in Step~\ref{step:urmila-in-naive}:
	for the second message, run $(y, \sigma)\leftarrow\cPNaiveStar{2}(pk, \rho)$.
	If $y\ne\bot$, then reply $y$;
	else, run the corresponding step of the dummy strategy and reply with its results.
	For the fourth message, if $y\ne\bot$, run and reply with $a\leftarrow\cPNaiveStar{4}(pk, c, \sigma)$;
	else, continue the dummy strategy.

	 In the following we fix an $x$. Let the distribution on $h$ specified in Step~\ref{step:naive1} of the protocol be $p_x(h)$. Define $\Pstarsub(x)$ as $\Pstar$'s response in Step~\ref{step:urmila-in-naive}. Note that we can view $\Pstarsub(x)$ as a prover strategy for \Cref{proto:urmila4}. By construction $\Pstarsub(x)$ passes testing round with overwhelming probability over $p_x(h)$, i.e. $\sum_h p_x(h) p_{h,T} =\negl(\lambda)$, where $p_{h,T}$ is $\Pstar$'s probability of getting accepted by the prover on the testing round on basis choice $h$. By \Cref{lem:urmila-binding} and Cauchy's inequality, there exists some $\rho$ such that  $\sum_h p_x(h) \norm{v_h -M_{XZ}(\rho, h)}_c = \negl(\lambda)$, where we use $\norm{A-B}_c=\alpha$ to denote that $A$ is $\alpha$-computational indistinguishable to $B$. Therefore $v= \sum_h p_x(h) v_h$ is computationally indistinguishable to $\sum_h p_x(h) M_{XZ}(\rho, h)$. Combining it with $\PiSamp$'s soundness (\Cref{QPIP1thm}), 
	we see that $(d', z')\leftarrow(\Pstar, \VNaive)(1^\lambda, 1^{1/\epsilon}, x)$  is $\eps$-computationally indistinguishable to $(d', z_{ideal}')$.

	Now we relate $(d', z')$ back to $(d, z)$.
	First, conditioned on that $\PNaiveStar$ aborts, since dummy strategy will be rejected with overwhelming probability in Hadamard round,
	we have $(d', z')$ is computationally indistinguishable to $(\Rej, \bot)=(d, z)$.
	On the other hand, conditioned on $\PNaiveStar$ not aborting, clearly $(d, z)=(d', z')$.
	So $(d, z)$ is computationally indistinguishable to $(d', z')$,
	which in turn is $O(\eps)$-computationally indistinguishable to $(d', z_{ideal}')$.
	Since $\norm{d-d'}_{tr}= O(\eps)$,
	 $(d, z_{ideal})$ is $O(\eps)$-computationally indistinguishable to $(d', z_{ideal}')$.
	Combining everything, we conclude that $(d, z)$ is $O(\eps)$-computationally indistinguishable to $(d, z_{ideal})$.
\end{proof}

\subsection{$\QPIP_0$ protocol for $\SampBQP$} \label{sec:qpip0}

We now introduce our $\QPIP_0$ protocol $\PiSampZ$ for $\SampBQP$.
It is essentially a $m$-fold parallel repetition of $\PiNaive$,
from which we uniformly randomly pick one copy to run Hadamard round to get our samples and run testing round on all other $m-1$ copies.
Intuitively, if the server wants to cheat by sending something not binding on some copy,
he will be caught when that copy is a testing round, which is with probability $1-1/m$.
This over-simplified analysis does not take into account that the server might create entanglement between the copies. Therefore, a more technically involved analysis is required.

In the description of our protocol below, we describe $\PiNaive$ and $\PiMeasure$ in details in order to introduce notations that we need in our analysis.

\begin{protocol}{$\QPIP_0$ protocol $\PiSampZ$ for the $\SampBQP$ problem $(D_x)_{x\in\set{0, 1}^*}$}
	\label{proto:QPIP0samp}

	Inputs:
	\begin{itemize}
		\item Security parameter $1^\lambda$ for $\lambda\in\bbN$.
		\item Accuracy parameter $1^{1/\eps}$ for the $\SampBQP$ problem.
		\item Input $x\in\zo^{\poly(\lambda)}$ for the $\SampBQP$ instance.
	\end{itemize}

	Ingredient: Let $m=O(1/\eps^2)$ be the number of parallel repetitions to run.
\bigskip 

	Protocol:
	\begin{enumerate}
		\item  The verifier  generates $m$ independent copies of basis choices $\vec{h}=(h_1,\ldots,h_m)$, where each copy is generated as in \stepref{naive1} of $\PiNaive$.
		\item The prover prepares $\rho^{\otimes m}$; each copy of $\rho$ is prepared as in \stepref{naive2} of $\PiNaive$.
		\item \label{step:urmila-in-qpip0-1} The verifier generates $m$ key pairs for $\PiMeasure$, $\vec{pk}=(pk_1,\ldots,pk_m)$ and $\vec{sk}=(sk_1,\ldots,sk_m)$, as in \stepref{measure1} of $\PiMeasure$.
			It sends $\vec{pk}$ to the prover.
		\item \label{step:urmila-in-qpip0-2}The prover generates $\vec{y}=(y_1,\ldots,y_m)$ and $\sigma$ as in \stepref{measure2} of $\PiMeasure$.
			It sends $\vec{y}$ to the verifier.
		\item \label{step:urmila-in-qpip0-3}The verifier samples $r\xleftarrow{\$}[m]$ which is the copy to run Hadamard round for.
			For $1\leq i\leq m$, if $i\ne r$ then set $c_i\leftarrow 0$, else set $c_i\leftarrow 1$.
			It sends $\vec{c}=(c_1,\ldots,c_m)$ to the prover.
		\item \label{step:urmila-in-qpip0-4}The prover generates $\vec{a}$ as in \stepref{measure4} of $\PiMeasure$, and sends it back to the verifier.
		\item The verifier computes the outcome for each round as in \stepref{naive-output} of $\PiNaive$.
			If any of the testing round copies are rejected, the verifier outputs $(\Rej, \bot)$.
			Else, it outputs the result from the Hadamard round copy.
	\end{enumerate}
\end{protocol}
By inspection, $\PiSampZ$ is a $\QPIP_0$ protocol for $\SampBQP$ with negligible completeness error.
To show that it is computationally sound, we first use the partition lemma from \cite{arXiv:ChiaChungYam19}.

Intuitively, the partition lemma says that for any cheating prover and for each copy $i\in[m]$, there exist two efficient ``projectors" \footnote{Actually they are not projectors, but for the simplicity of this discussion let's assume they are.} $G_{0,i}$ and $G_{1,i}$ in the prover's internal space with $G_{0,i}+G_{1,i} \approx Id$. $G_{0,i}$ and $G_{1,i}$ splits up the prover's residual internal state after sending back his first message.
$G_{0,i}$ intuitively represents the subspace where the prover does not knows the answer to the testing round on the $i$-th copy, while $G_{1,i}$ represents the subspace where the prover does. Note that the prover is using a single internal space for all copies, and every $G_{0,i}$ and every $G_{1,i}$ is acting on this single internal space. 
By using this partition lemma iteratively, we can decompose the prover's internal state $\ket{\psi}$ into sum of subnormalized states.
First we apply it to the first copy, writing $\ket{\psi}=G_{0,1}\ket{\psi}+G_{1,1}\ket{\psi} \equiv \ket{\psi_0}+\ket{\psi_1}$.
The component $\ket{\psi_0}$ would then get rejected as long as the first copy is chosen as a testing round,
which occurs with pretty high probability.
More precisely, the output corresponding to $\ket{\psi_0}$ is $1/m$-close to the ideal distribution that just rejects all the time.
On the other hand, $\ket{\psi_1}$ is now binding on the first copy;
we now similarly apply the partition lemma of the second copy to $\ket{\psi_1}$.
We write $\ket{\psi_1}=G_{0,2}\ket{\psi_1}+G_{1,2}\ket{\psi_1}\equiv \ket{\psi_{10}}+\ket{\psi_{11}}$, and apply the same argument about $\ket{\psi_{10}}$ and $\ket{\psi_{11}}$.
We then continue to decompose $\ket{\psi_{11}}=\ket{\psi_{110}}+\ket{\psi_{111}}$ and so on, until we reach the last copy and obtain $\ket{\psi_{1^m}}$.
Intuitively, the $\ket{\psi_{1^m}}$ term represents the ``good" component where the prover knows the answer to every testing round and therefore has high accept probability. Therefore, $\ket{\psi_{1^m}}$ also satisfies some binding property,
so the verifier should obtain a measurement result of some state on the Hadamard round copy,
and the analysis from the $\QPIP_1$ protocol $\PiSamp$ follows.

However, the intuition that $\ket{\psi_{1^m}}$ is binding to every Hadamard round is incorrect. As $G_{1,i}$ does not commute with $G_{1,j}$, $\ket{\psi_{1^m}}$ is unfortunately only binding for the $m$-th copy.
To solve this problem, we start with a pointwise argument and fix the Hadamard round on the $i$-th copy where $\ket{\psi_{1^i}}$ is binding,
and show that the corresponding output is $O(\norm{\ket{\psi_{1^{i-1}0}}})$-close to ideal.
We can later average out this error over the different choices of $i$, since not all $\norm{\ket{\psi_{1^{i-1}0}}}$ can be large at the same time. Another way to see this issue is to notice that we are partitioning a quantum state, not probability events, so there are some inconsistencies between our intuition and calculation. Indeed, the error we get in the end is $O(\sqrt{1/m})$ instead of the $O(1/m)$ we expected.

Also a careful reader might have noticed that the prover's space don't always decompose cleanly into parts that the verifier either rejects or accepts with high probability, as there might be some states that is accepted with mediocre probability. As in \cite{arXiv:ChiaChungYam19}, we solve this by splitting the space into parts that are accepted with probability higher or lower than a small threshold $\gamma$ and applying Marriott-Watrous~\cite{marriott2005quantum} amplification to boost the accept probability if it is bigger than $\gamma$, getting a corresponding amplified prover action $\ext$. However, states with accept probability really close to the threshold $\gamma$ can not be classified, so we average over randomly chosen $\gamma$ to have $G_{0,i}+G_{1,i} \approx Id$. Now we give a formal description of the partition lemma.

\begin{lemma}[partition lemma; revision of Lemma 3.5 of \cite{arXiv:ChiaChungYam19}\footnote{$G_{0}$ and $G_{1}$ of this version are created from doing $G$ of \cite{arXiv:ChiaChungYam19} and post-selecting on the $ph,th,in$ register being $0^t01$ or $0^t11$ then discard $ph,th,in$. Property~\ref{property:partition-err} corresponds to Property~1. Property~\ref{property:partition-testing} corresponds to Property~4, with $2^{m-1}$ changes to $m-1$ because we only have $m$ possible choices of $\vec{c}$. Property~\ref{property:partition-binding} corresponds to Property~5. Property~\ref{property-partition-norm-sum} comes from the fact that $G_0$ and $G_1$ are post-selections of orthogonal results of the same $G$.}]\label{lem:partition2}
	Let $\lambda$ be the security parameter, and $\gamma_0 \in[0,1]$ and $T\in \mathbb{N}$ be parameters that will be related to the randomly-chosen threshold $\gamma$.
	Let $(U_0,U)$ be a prover's strategy in a $m$-fold parallel repetition of $\PiMeasure$\footnote{A $m$-fold parallel repetition of $\PiMeasure$ is running step~\ref{step:urmila-in-qpip0-1}~\ref{step:urmila-in-qpip0-2}~\ref{step:urmila-in-qpip0-3}~\ref{step:urmila-in-qpip0-4} of \Cref{proto:QPIP0samp} with verifier input $\vec{h}$ and prover input $\rho^{\otimes n}$, followed by an output step where the verifier rejects if any of the $m-1$ testing round copies is rejected, otherwise outputs the result of the Hadamard round copy.}, where $U_0$ is how the prover generates $\vec{y}$ on the second message, and $U$ is how the prover generates $\vec{a}$ on the fourth message. Let $H_{\regX,\regZ}$ be the Hilbert space of the prover's internal calculation.
	Denote the string $0^{i-1}10^{m-i} \in \zo^m $ as $e_i$, which corresponds to doing Hadamard round on the $i$-th copy and testing round on all others.

	For all $i\in[m]$, $\gamma \in \L\{\frac{\gamma_0}{T},\frac{2\gamma_0}{T},\dots,\frac{T\gamma_0}{T}\R\}$, there exist two $\poly(1/\gamma_0,T,\lambda)$-time quantum circuit with post selection\footnote{A quantum circuit with post selection is composed of unitary gates followed by a post selection on some measurement outcome on ancilla qubits, so it produces a subnormalized state, where the amplitude square of the output state is the probability of post selection.} $G_{0,i,\gamma}$ and $G_{1,i,\gamma}$ such that for all (possibly sub-normalized)  quantum states $\ket{\psi}_{\regX,\regZ}\in  H_{\regX,\regZ}$,  properties \ref{property:partition-err}~\ref{property:partition-testing}~\ref{property:partition-binding}~\ref{property-partition-norm-sum}, to be described later, are satisfied. Before we describe the properties, we introduce the following notations:  

	\begin{align}
	\label{eq:psi0}
	\ket{\psi_{0,i,\gamma}}_{\regX,\regZ}
	\defeq&
	G_{0,i,\gamma}\ket{\psi}_{\regX,\regZ}, \\ 
	\label{eq:psi1}	
	\ket{\psi_{1,i,\gamma}}_{\regX,\regZ}
	 \defeq&
	G_{1,i,\gamma}\ket{\psi}_{\regX,\regZ}
	,  \\
	\label{eq:psierr}
	\ket{\psi_{err,i,\gamma}}_{\regX,\regZ}
	\defeq&
	\ket{\psi}_{\regX,\regZ} -\ket{\psi_{0,i,\gamma}}_{\regX,\regZ}- \ket{\psi_{1,i,\gamma}}_{\regX,\regZ}
	.
	\end{align}

	Note that $G_{0,i,\gamma}$ and $G_{1,i,\gamma}$ has failure probabilities, and this is reflected by the fact that $\ket{\psi_{0,i,\gamma}}_{\regX,\regZ}$ and $\ket{\psi_{1,i,\gamma}}_{\regX,\regZ}$ are  sub-normalized. $G_{0,i,\gamma}$ and $G_{1,i,\gamma}$ depend on $(U_0,U)$ and $\vec{pk},\vec{y}$.

	The following properties are satisfied for all $i\in[m]$:
	\begin{enumerate}
		\item \label{property:partition-err}  $$\E_{\gamma}\|\ket{\psi_{err,i,\gamma}}_{\regX,\regZ}\|^2 \leq \frac{6}{T}+\negl(\lambda),$$

			where the averaged is over uniformly sampled $\gamma$. This also implies
			\begin{align}
				\E_{\gamma}\|\ket{\psi_{err,i,\gamma}}_{\regX,\regZ}\| \leq \sqrt{\frac{6}{T}}+\negl(\lambda)
			\end{align}
			by Cauchy's inequality.

		\item \label{property:partition-testing}
			For all $\vec{pk}$, $\vec{y}$, $\gamma$, and  $j\neq i$, we have
			\begin{align}
				\norm{ P_{acc, i} \circ U\frac{\ket{e_j}_{\regC}\ket{\psi_{0,i,\gamma}}_{\regX,\regZ}}{\|\ket{\psi_{0,i,\gamma}}_{\regX,\regZ}\|}}^2 \leq (m-1)\gamma_0+\negl(\lambda),
			\end{align}
			where $P_{acc, i}$ are projector to the states that $i$-th testing round accepts with $pk_i,y_i$, including the last measurement the prover did before sending $\vec{a}$.  This means that $\ket{\psi_{0,i,\gamma}}$ is rejected by the $i$-th testing round with high probability.
		\item \label{property:partition-binding}
			For all $\vec{pk}$, $\vec{y}$, $\gamma$, and $j\neq i$, there exists an efficient quantum algorithm $\ext_i$ such that
			\begin{align}
				\norm{P_{acc, i} \circ \ext_i\left(\frac{\ket{e_j}_{\regC}\ket{\psi_{1,i,\gamma}}_{\regX,\regZ}}{\|\ket{\psi_{1,i,\gamma}}_{\regX,\regZ}\|}\right)}^2 =1-\negl(\lambda).
			\end{align}

			This will imply that $\ket{\psi_{1,i,\gamma}}$ is binding to the $i$-th Hadamard round.

		\item \label{property-partition-norm-sum}
			For all $\gamma$,
			\begin{align}
				\norm{\ket{\psi_{0,i,\gamma}}_{\regX,\regZ}}^2+ \norm{\ket{\psi_{1,i,\gamma}}_{\regX,\regZ}}^2 \leq  \norm{\ket{\psi}_{\regX,\regZ}}^2.
			\end{align}
	\end{enumerate}
\end{lemma}

Note that in property~\ref{property:partition-binding}, we are using $\ext_i$ instead of $U$ because we use amplitude amplification to boost the success probability. 

We now decompose the prover's internal state by using \Cref{lem:partition2} iteratively.
Let $\ket{\psi}$ be the state the prover holds before he receives $\vec{c}$;
we denote the corresponding Hilbert space as $H_{\regX,\regZ}$.
For all $k \in [m]$, $d\in \zo^k$, $\gamma=(\gamma_1, \ldots, \gamma_k)$ where each $\gamma_j\in\set{\frac{\gamma_0}{T},\frac{2\gamma_0}{T},\dots,\frac{T\gamma_0}{T}}$,  
and $\ket{\psi} \in H_{\regX,\regZ}$, define $$\ket{\psi_{d,\gamma}}\defeq G_{d_k,k,\gamma_k}\ldots G_{d_2,2,\gamma_2} G_{d_1,1,\gamma_1} \ket{\psi}.$$
For all $i\in[m]$, we then decompose $\ket{\psi}$ into
\begin{equation}
	\label{eq:partition-string}
	\ket{\psi}=\sum_{j=0}^{i-1} \ket{\psi_{1^j0,\gamma}} +\ket{\psi_{1^i,\gamma}} +\sum_{j=1}^{i}\ket{\psi_{err,j,\gamma}}
\end{equation}
by using \Cref{eq:psi0,eq:psi1,eq:psierr} repeatedly,  where $\ket{\psi_{err,i,\gamma}}$ denotes the error state from decomposing $\ket{\psi_{1^{i-1},\gamma}}$.

We denote the projector in $H_{\regX,\regZ}$ corresponding to outputting string $z$ when doing Hadamard on $i$-th copy as $P_{acc,-i,z}$.
Note that $P_{acc,-i,z}$ also depends on $\vec{pk}, \vec{y}$, and $(sk_i, h_i)$ since it includes the measurement the prover did before sending $\vec{a}$,  verifier's checking on $(m-1)$ copies of testing rounds, and  the verifier's final computation from $(sk_i,h_i,y_i,a_i)$. $P_{acc,-i,z}$ is a projector because it only involves the standard basis measurements to get $a$ and classical post-processing of the verifiers. Also note that  $P_{acc,-i,z} P_{acc,-i,z'}=0$ for all $z\neq z'$, and $\sum_z P_{acc,-i,z} =\Pi_{j \neq i} P_{acc,j}\leq Id$.

We denote the string $0^{i-1}10^{m-i} \in \zo^m$ as $e_i$. The output string corresponding to $\ket{\psi} \in H_{\regX,\regZ}$ when $c=e_i$ is then
\begin{equation}
	\label{eq:zi-def}
	z_i\defeq \E_{pk,y} \sum_z \norm{P_{acc,-i,z} U\ket{e_i,\psi}}^2\proj{z},
\end{equation}
 where $\ket{e_i,\psi}=\ket{e_i}_\regC\ket{\psi}_{\regX,\regZ}$ and $U$ is the unitary the prover applies on the last round.
Note that we have averaged over $\vec{pk}, \vec{y}$ where as previously everything has fixed $\vec{pk}$ and $\vec{y}$.

By Property~\ref{property:partition-testing} of \Cref{lem:partition2},
it clearly follows that 
\begin{cor}
	\label{lem:partition-testing}
	For all $\gamma\in\set{\frac{\gamma_0}{T},\frac{2\gamma_0}{T},\dots,\frac{T\gamma_0}{T}}$, and all $i,j\in[m]$ such that $j<i-1$, we have
	$$\norm{\sum_z P_{acc,-i,z} U \ket{e_i, \psi_{1^j0,\gamma}}}^2\leq (m-1)\gamma_0+\negl(n).$$
\end{cor}

Now we define
\begin{equation}
	\label{eq:zgoodi-def}
	z_{good, i}=\E_{\gamma, pk, y} \sum_z \norm{P_{acc,-i,z} U\ket{e_i,\psi_{1^{i-1}1,\gamma}}}^2\proj{z}
\end{equation}
as the output corresponding to a component that would pass the $i$-th testing rounds.
We will show that it is $O(\norm{\ket{\psi_{1^{i-1}0}}})$-close to $z_i$.
Before doing so, we present a technical lemma.

\begin{lemma}\label{lem:samp-tech-2}
	For any state $\ket{\psi}$,  $\ket{\phi}$ and projectors $\{P_z\}$ such that $P_z P_{z'} =0 $ for all $z\neq z'$, we have
	$$  \sum_z |\vev{\psi|P_z|\phi}| \leq \sqrt{\norm{\sum_z P_z\ket{\psi}}^2 } \sqrt{\norm{\sum_z P_z\ket{\phi}}^2 }. $$
\end{lemma}
\begin{proof}
	\begin{align}
		\sum_z |\vev{\psi|P_z|\phi}| =&\sum_z|\vev{\psi|P_zP_z|\phi}| \nn \\
		\leq& \sum_z \norm{\bra{\psi}P_z} \norm{ P_z\ket{\phi}} \nn \\
		\leq&  \sqrt{\sum_z \norm{P_z\ket{\psi}}^2} \sqrt{\sum_z\norm{P_z\ket{\phi}}^2} \nn \\
		\leq& \sqrt{\norm{\sum_z P_z\ket{\psi}}^2 } \sqrt{\norm{\sum_z P_z\ket{\phi}}^2 } \nn,
	\end{align}
	where we used Cauchy's inequality on the second and third line and $P_z P_{z'} =0 $ on the fourth line.
\end{proof}

\begin{cor}\label{lem:samp-tech}
	For any state $\ket{\psi}$,  $\ket{\phi}$ and projectors $\{P_z\}$ such that $\sum_z P_z \leq Id$ and $P_z P_{z'} =0 $ for all $z\neq z'$, we have
	$$  \sum_z |\vev{\psi|P_z|\phi}| \leq \norm{\psi}\norm{\phi}. $$
\end{cor}

Now we can estimate $z_i$ using $z_{good, i}$, with errors on the orders of $\norm{\ket{\psi_{1^{i-1}0}}}$.
This error might not be small in general,
but we can average it out later by considering uniformly random $i\in[m]$.
The analysis is tedious but straightforward;
we simply expand $z_i$ and bound the terms that are not $z_{good, i}$.

\begin{lemma}
	\label{thm:zi-zgoodi}
	\begin{align*}
	\tr\abs{z_i-z_{good, i}}\leq&\E_{pk, y, \gamma}\L[\norm{\ket{\psi_{1^{i-1}0,\gamma}}}^2+2\norm{\ket{\psi_{1^{i-1}0,\gamma}}}\R]\\
	&+O\L(\frac{m^2}{\sqrt T}+m\sqrt{(m-1)\gamma_0}\R).
	\end{align*}
\end{lemma}

\begin{proof}[\Cref{thm:zi-zgoodi}]
	We take expectation of \Cref{eq:partition-string} over $\gamma$
	$$\ket{\psi}=\E_{\gamma}\left[
		\sum_{j=0}^{i-1} \ket{\psi_{1^j0,\gamma}} +\ket{\psi_{1^i,\gamma}} +\sum_{j=1}^{i}\ket{\psi_{err,j,\gamma}}
	\right],$$
	and expand $z_i$ from \Cref{eq:zi-def} as
	\begin{align}
		z_i &= z_{good,i}+ \E_{pk, y, \gamma} \sum_z \L[\sum_{k=0}^{i-1} \bra{\psi_{1^k0,\gamma}}U^\dag  P_{acc,-i,z}U   \sum_{j=0}^{i-1} \ket{\psi_{1^j0,\gamma}} \R. \nn \\
		&+
		\sum_{k=0}^{i-1} \bra{\psi_{1^k0,\gamma}}U^\dag  P_{acc,-i,z}U \ket{\psi_{1^i,\gamma}} +\sum_{k=0}^{i-1} \bra{\psi_{1^k0,\gamma}}U^\dag  P_{acc,-i,z}U\sum_{j=1}^{i}\ket{\psi_{err,j,\gamma}} \nn \\
		&+\bra{\psi_{1^i,\gamma}} U^\dag  P_{acc,-i,z}U \sum_{j=0}^{i-1} \ket{\psi_{1^j0,\gamma}} +\bra{\psi_{1^i,\gamma}} U^\dag  P_{acc,-i,z}U \sum_{j=1}^{i}\ket{\psi_{err,j,\gamma}}
		\nn \\
		&+ \sum_{k=1}^{i}\bra{\psi_{err,k,\gamma}} U^\dag  P_{acc,-i,z}U  \sum_{j=0}^{i-1} \ket{\psi_{1^j0,\gamma}} + \sum_{k=1}^{i}\bra{\psi_{err,k,\gamma}} U^\dag  P_{acc,-i,z}U \ket{\psi_{1^i,\gamma}}
		\nn \\
		&\L.    +\sum_{k=1}^{i}\bra{\psi_{err,k,\gamma}} U^\dag  P_{acc,-i,z}U \sum_{j=1}^{i}\ket{\psi_{err,j,\gamma}} \R] \proj{z} , \nn     
	\end{align}
	where we omitted writing out $e_i$.
	Therefore we have
	\begin{align*}
		\tr|z_i-z_{good,i}|\leq \E_{pk, y, \gamma} \sum_z &\L[ \sum_{k=0}^{i-1} \sum_{j=0}^{i-1} \L| \bra{\psi_{1^k0,\gamma}}U^\dag  P_{acc,-i,z}U \ket{\psi_{1^j0,\gamma}} \R|\R.\\
		&+
		2 \sum_{k=0}^{i-1} \L|\bra{\psi_{1^k0,\gamma}}U^\dag  P_{acc,-i,z}U \ket{\psi_{1^i,\gamma}} \R| \\
		&+ 2 \sum_{k=0}^{i-1}\sum_{j=1}^{i}\L| \bra{\psi_{1^k0,\gamma}}U^\dag  P_{acc,-i,z}U\ket{\psi_{err,j,\gamma}}\R| \\   
		&+2 \sum_{j=1}^{i}\L|\bra{\psi_{1^i,\gamma}} U^\dag  P_{acc,-i,z}U \ket{\psi_{err,j,\gamma}}\R| \\
		&+\L. \sum_{k=1}^{i}\sum_{j=1}^{i}\L| \bra{\psi_{err,k,\gamma}} U^\dag  P_{acc,-i,z}U \ket{\psi_{err,j,\gamma}}\R| \R]
	\end{align*}
	by the triangle inequality.
	The last three error terms sum to $O\L(\frac{m^2}{\sqrt{T}}\R)$ by \Cref{lem:samp-tech} and property~\ref{property:partition-err} of \Cref{lem:partition2}.
	As for the first two terms, by \Cref{lem:samp-tech-2} and \Cref{lem:partition-testing}, we see that
	\begin{align*}
		\sum_z \sum_{k=0}^{i-1}\sum_{j=0}^{i-1}
		&\abs{\bra{\psi_{1^k0,\gamma}}U^\dag  P_{acc,-i,z}U \ket{\psi_{1^j0,\gamma}}} \\
		&\leq\sum_z \abs{\bra{\psi_{1^{i-1}0,\gamma}}U^\dag  P_{acc,-i,z}U \ket{\psi_{1^{i-1}0,\gamma}}} + O\L(m^2(m-1)\gamma_0\R) \\
		&\leq\norm{\ket{\psi_{1^{i-1}0,\gamma}}}^2 + O\L(m^2(m-1)\gamma_0\R)
	\end{align*}
	and similarly
	\begin{align*}
		\sum_z\sum_{k=0}^{i-1}
		&\abs{\bra{\psi_{1^k0,\gamma}}U^\dag  P_{acc,-i,z}U \ket{\psi_{1^i,\gamma}}}\\
		&\leq\sum_z\abs{\bra{\psi_{1^{i-1}0,\gamma}}U^\dag  P_{acc,-i,z}U \ket{\psi_{1^i,\gamma}}}+O\L(m\sqrt{(m-1)\gamma_0}\R)\\
		&\leq\norm{\ket{\psi_{1^i,\gamma}}}+O\L(m\sqrt{(m-1)\gamma_0}\R).
	\end{align*}
\end{proof}

Now let $z_{true}$, as a mixed state, be the correct sample of the $\SampBQP$ instance $D_x$,
and let $z_{ideal, i}=\tr(z_{good, i})z_{true}$.
We show that $z_{ideal, i}$ is close to $z_{good, i}$.
\begin{lemma}
	\label{thm:zgood-zideal}
	$z_{good, i}$ is $O(\eps)$-computationally indistinguishable to $z_{ideal, i}$,
	where $\eps\in\bbR$ is the accuracy parameter picked earlier in $\PiSampZ$.
\end{lemma}
\begin{proof}[\Cref{thm:zgood-zideal}]
	For every $i\in [m]$ and every prover strategy $(U_0,U)$ for $\PiSampZ$, consider the following composite strategy, $\Picomp{i}$, as a prover for $\PiNaive$. Note that a prover only interacts with the verifier in Step~\ref{step:urmila-in-naive} of $\PiNaive$ where $\PiMeasure$ is run, so we describe a prover's action in terms of the four rounds of communication in $\PiMeasure$. 

	$\Picomp{i}$ tries to run $U_0$ by taking the verifier's input as the input to the $i$-th copy of $\PiMeasure$ in $\PiSampZ$ and simulating other $m-1$ copies by himself. The prover then picks a uniformly random $\gamma$ and  tries to generate $\ket{\psi_{1^{i-1}1,\gamma}}$ by applying $G_{i,1,\gamma}G_{i-1,1,\gamma} \cdots G_{2,1,\gamma}G_{1,1,\gamma}$. This can be efficiently done because of \Cref{lem:partition2} and our choice of $\gamma_0$ and $T$ in \Cref{thm:qpip0}. If the prover fails to generate $\ket{\psi_{1^{i-1}1,\gamma}}$, he throws out everything and aborts by sending $\bot$ back.   On the fourth round,  If it's a testing round the prover reply with the $i$-th register of $\ext_i\left(\frac{\ket{e_j}_{\regC}\ket{\psi_{1,i,\gamma}}_{\regX,\regZ}}{\|\ket{\psi_1}_{\regX,\regZ}\|}\right)$, where $\ext_i$ is specified in property~\ref{property:partition-binding} of Lemma~\ref{lem:partition2}. If it's the Hadamard round  the prover  runs $U$ and checks whether every copy except the $i$-th copy would be accepted. If all $m-1$ copies are accepted, he replies with the $i$-th copy, otherwise reply $\bot$.
	
	Denote the result we would get in the Hadamard round by $z_{composite,i}$. By construction, when $G_{i,1,\gamma}\ldots G_{1,1,\gamma}$ succeeded, the corresponding output would be $z_{good,i}$. Also note that this is the only case where the verifier won't reject, so $z_{composite,i}=z_{good,i}$.

	In the testing round, by property~\ref{property:partition-binding} of ~\Cref{lem:partition2}, the above strategy is accepted with probability $1-\negl(n)$ when the prover didn't abort.
	Since the prover's strategy is also efficient, by ~\Cref{lem:naive-qpip0-binding},
	$z_{composite,i}$ is $O(\eps)$-computationally indistinguishable to $z_{ideal, i}$. 
\end{proof}

Now we try to put together all $i\in [m]$. First let
$$z=\frac{1}{m} \sum_i z_i= \frac{1}{m} \sum_i \sum_z \proj{z} \cdot \vev{e_i,\psi|U^\dag P_{acc,-i,z} U|e_i,\psi},$$
which is the output distribution of $\PiSampZ$.
We also define the following accordingly:
$$z_{good}\defeq \frac{1}{m}\sum_i z_{good,i,}$$
$$z_{ideal}\defeq \frac{1}{m}\sum_i z_{ideal,i}.$$
Notice that $z_{ideal}$ is some ideal output distribution, which might not have the same accept probability as $z$.

\begin{theorem}\label{thm:qpip0} 
	Under the QLWE assumption, $\PiSampZ$ is a protocol for the $\SampBQP$ problem $(D_x)_{x\in\set{0,1}^*}$  with negligible completeness error and is computationally sound.\footnote{The soundness and completeness of a $\SampBQP$ protocol is defined in \Cref{dfn:stats-secure-proto-sampbqp}}
	
\end{theorem}
\begin{proof}
	Completeness is trivial. In the following we prove the soundness.
	
	By Property~\ref{property-partition-norm-sum} of Lemma~\ref{lem:partition2}, we have
	\begin{align} \label{eq:bad-term-sum}
		\norm{\ket{\psi}}^2 \geq& \norm{\ket{\psi_{0,\gamma}}}^2+\norm{\ket{\psi_{1,\gamma}}}^2 \nn \\
		\geq& \norm{\ket{\psi_{0,\gamma}}}^2+
		\norm{\ket{\psi_{10,\gamma}}}^2+ \norm{\ket{\psi_{11,\gamma}}}^2 \nn \\
		\geq& \norm{\ket{\psi_{0,\gamma}}}^2+
		\norm{\ket{\psi_{10,\gamma}}}^2+ \norm{\ket{\psi_{110,\gamma}}}^2 +\cdots  \nn \\
		&+ \norm{\ket{\psi_{1^{m-1}0,\gamma}}}^2+ \norm{\ket{\psi_{1^{m-1}1,\gamma}}}^2.
	\end{align}

	We have
	\begin{align} \label{eq:z-z-good}
		\tr|z-z_{good}| =& \tr\L|\frac{1}{m}\sum_i (z_i-z_{good,i})\R| \nn \\
		\leq&  \frac{1}{m}\sum_i\tr| (z_i-z_{good,i})| \nn \\
		\leq&  \frac{1}{m}\sum_i\L[\E_{pk, y, \gamma}\L[\norm{\ket{\psi_{1^{i-1}0,\gamma}}}^2+ 2\norm{\ket{\psi_{1^{i-1}0,\gamma}}}\R] \R.\nn \\
		&+ \L. O\L(\frac{m^2}{\sqrt T}+m\sqrt{(m-1)\gamma_0}\R)\R] \nn \\
		\leq&  \frac{1}{m}+ 2\frac{1}{\sqrt m}+O\L(\frac{m^2}{\sqrt T}+m\sqrt{(m-1)\gamma_0}\R) \nn \\ 
		=&O\L( \frac{1}{\sqrt m}+\frac{m^2}{\sqrt T}+m\sqrt{(m-1)\gamma_0}\R),  
	\end{align}
	where we used triangle inequality on the second line, \Cref{thm:zi-zgoodi} on the third line, Equation~\ref{eq:bad-term-sum} and Cauchy's inequality on the fourth line.
	Set $m=O(1/\eps^2), T=O(1/\eps^2),\gamma_0=\eps^8$. Combining \Cref{thm:zgood-zideal} and \Cref{eq:z-z-good} by triangle inequality, we have $z$ is $O(\epsilon)$-computationally indistinguishable to $z_{ideal}$. Therefore, $(d,z)$  $O(\epsilon)$-computationally indistinguishable to $(d,z_{ideal})$.
\end{proof}

~\Cref{thm:qpip0-informal} follows as a corollary.

\section{Generic Blindness Protocol Compiler for $\QPIP_0$}
\label{sec:BlindBQP2}

In this section, we present a generic protocol compiler that compiles any $\QPIP_0$ protocol $\Pi = (P, V)$ (with an arbitrary number of rounds) to a protocol $\Piblind = (\Pblind, \Vblind)$ that achieve blindness while preserving the completeness, soundness, and round complexity. At a high-level, the idea is simple: we simply run the original protocol under a quantum homomorphic encryption $\QHE$ with the verifier's key. Intuitively, this allows the prover to compute his next message under encryption without learning the underlying verifier's message, and hence achieves blindness while preserving the properties of the original protocol.

However, several issues need to be taking care to make the idea work. First, since the verifier is classical, we need the quantum homomorphic encryption scheme $\QHE$ to be \emph{classical friendly} as defined in Definition~\ref{def:classical-friendly}. Namely, the key generation algorithm and the encryption algorithm for classical messages should be classical, and when the underlying message is classical, the ciphertext (potentially from homomorphic evaluation) and the decryption algorithm should be classical as well. Fortunately, the quantum homomorphic encryption scheme of Mahadev~\cite{mahadev_qfhe} and Brakerski~\cite{brakerski_qfhe} are classical friendly. Moreover, Brakerski's scheme requires a weaker QLWE assumption, where the modulus is polynomial instead of super-polynomial.

A more subtle issue is to preserve the soundness. Intuitively, the soundness holds since the execution of $\Piblind$ simulates the execution of $\Pi$, and hence the soundness of $\Pi$ implies the soundness of $\Piblind$. However, to see the subtle issue, let us consider the following naive compiler that uses a single key: In $\Piblind$,  the verifier $V$ initially generates a pair $\QHE$ key $(pk, sk)$, sends $pk$ and encrypted input $\QEnc(pk,x)$ to $P$. Then they run $\Pi$ under encryption with this key, where both of them use homomorphic evaluation to compute their next message.

There are two reasons that the compiled protocol $\Piblind$ may not be sound (or even not blind). First, in general, the $\QHE$ scheme may not have \emph{circuit privacy}; namely, the homomorphic evaluation may leak information about the circuit being evaluated. Since the verifier computes his next message using homomorphic evaluation, a cheating prover $\Pblindstar$ seeing the homomorphically evaluated ciphertext of the verifier's message may learn information about the verifier's next message circuit, which may contain information about the secret input $x$ or help $\Pblindstar$ to break the soundness. Second, $\Pblindstar$ may send invalid ciphertexts to $V$, so the execution of $\Piblind$ may not simulate a valid execution of $\Pi$.

To resolve the issue, we let the verifier switch to a fresh new key for each round of the protocol. For example, when the prover $\Pblind$ returns the ciphertext of his first message, the verifier $\Vblind$ decrypts the ciphertext, computes his next message (in the clear), and then encrypt it using a fresh key $pk'$ and sends it to $\Pblind$.
Note that a fresh key pair is necessary here to ensure blindness, as decrypting uses information from the secret key.
Since the verifier $\Vblind$ only  sends fresh ciphertexts to $\Pblind$, this avoids the issue of circuit privacy. Additionally, to allow $\Pblind$ to homomorphically evaluate its next message, $\Vblind$ needs to encrypt the previous secret key $sk$ under the new public key $pk'$ and send it along with $pk'$ to $\Pblind$. This allows the prover to homomorphically convert ciphertexts under key $pk$ to ciphertexts under key $pk'$. By doing so, we show that for any cheating prover $\Pblindstar$, the interaction $(\Pblindstar, \Vblind)$ indeed simulates a valid interaction of $(\Pstar, V)$ for some cheating $\Pstar$, and hence the soundness of $\Pi$ implies the soundness of the compiled protocol. Finally, for the issue of the prover sending invalid ciphertexts, we note that this is not an issue if the decryption never fails, which can be achieved by simply let the decryption algorithm output a default dummy message (e.g., $0$) when it fails. 

We note that the idea of running the protocol under homomorphic encryptions is used in~\cite{KMThesis} in a classical setting, but for a different purpose of making the protocol ``computationally simulatable'' in their context.

We proceed to present our compiler. We start by introducing the notation of a $\QPIP_0$ protocol $\Pi$ as follows.

\begin{protocol}{$\QPIP_0$ protocol $\Pi=(P, V)(x)$ where only the verifier receives outputs}
    
    Common inputs\footnote{For the sake of simplicity, we omit accuracy parameter $\eps$ where it exists}:
    \begin{itemize}
        \item Security parameter $1^\lambda$ where $\lambda\in\bbN$
        \item A classical input $x\in\zo^{\poly(\lambda)}$
    \end{itemize}

    Protocol:
    \begin{enumerate}
        \item $V$ generates $(v_1, st_{V, 1})\leftarrow\cV_1(1^\lambda, x)$ and sends $v_1$ to the prover.
        \item $P$ generates $(p_1, st_{P, 1})\leftarrow\cP_1(1^\lambda, v_1, x)$ and sends $p_1$ to the verifier.
        \item for $t=2,\ldots,T$:
        \begin{enumerate}
            \item $V$ generates $(v_t, st_{V, t})\leftarrow\cV_t(p_{t-1}, st_{V, t-1})$ and sends $v_t$ to the prover.
            \item $P$ generates $(p_t, st_{P, t})\leftarrow\cP_t(v_t, st_{P, t-1})$ and sends $p_t$ to the verifier.
        \end{enumerate}
        \item $V$ computes its output $o\leftarrow\cV_{out}(p_T, st_{V, T})$.
    \end{enumerate}

\end{protocol}

We compile the above protocol to achieve blindness as follows.
For notation, when there are many sets of $\QHE$ keys in play at the same time,
we use $\ctx{x}{}{i}$ to denote $x$ encrypted under $pk_i$.

\begin{protocol}{Blind $\QPIP_0$ protocol $\Piblind=(\Pblind, \Vblind(x))$ corresponding to $\Pi_0$}
    Inputs:
    \begin{itemize}
        \item Common input: Security parameter $1^\lambda$ where $\lambda\in\bbN$
        \item Verifier's input: $x\in\zo^{\poly(\lambda)}$
    \end{itemize}

    Ingredients:
    \begin{itemize}
        \item Let $L$ be the maximum circuit depth of $\cP_t$.
    \end{itemize}

    Protocol:
    \begin{enumerate}
        \item $\Vblind$ generates $(v_1, st_{V, 1})\leftarrow\cV_1(1^\lambda, x)$.
            Then it generates $(pk_1, sk_1)\leftarrow\QGen(1^\lambda, 1^L)$,
            and encrypts $\ctx{x}{}{1}\leftarrow\QEnc(pk_1, x)$ and $\ctx{v}{1}{1}\leftarrow\QEnc(pk_1, v_1)$.
            It sends $pk_1$, $\ctx{x}{}{1}$, and $\ctx{v}{1}{1}$ to the prover.
        \item $\Pblind$ generates $(\ctx{p}{1}{1}, \ctx{st}{P, 1}{1})\leftarrow\cPblind{1}(1^\lambda, \ctx{v}{1}{1}, \ctx{x}{}{1})$
            by evaluating
            $(\ctx{p}{1}{1}, \ctx{st}{P, 1}{1})\leftarrow  \QEval(pk, \cP_1,$ \  $\QEnc(pk_1, 1^\lambda), \ctx{v}{1}{1}, \ctx{x}{}{1})$.
            It sends $\ctx{p}{1}{1}$ to the verifier.
        \item for $t=2,\ldots,T$:
        \begin{enumerate}
            \item $\Vblind$ decrypts the prover's last message by $p_{t-1}\leftarrow\QDec(sk_{t-1}, \ctx{p}{t-1}{t-1})$,
                then generates $(v_t, st_{V, t})\leftarrow\cV_t(p_{t-1}, st_{V, t-1})$.
                Then it generates $(pk_t, sk_t)\leftarrow\QGen(1^\lambda, 1^L)$,
                and produces encryptions $\ctx{v}{t}{t}\leftarrow\QEnc(pk_t, v_t)$ and $\ctx{sk}{t-1}{t}\leftarrow\QEnc(pk_t, sk_{t-1})$.
                It sends $pk_t$, $\ctx{v}{t}{t}$, and $\ctx{sk}{t-1}{t}$ to the prover.
            \item $\Pblind$ generates $(\ctx{p}{t}{t}, \ctx{st}{P, t}{t})\leftarrow\cPblind{t}(\ctx{v}{t}{t}, \ctx{sk}{t-1}{t}, \ctx{st}{P, t-1}{t-1})$
                by first switching its encryption key;
                that is, it encrypts its state under the new key by $\ctx{st}{P, t-1}{t-1, t}\leftarrow\QEnc(pk_t, \ctx{st}{P, t-1}{t-1}))$,
                then homomorphically decrypts the old encryption by
                $\ctx{st}{P, t-1}{t}\leftarrow\QEval(pk_t, \QDec,$ \ $\ctx{sk}{t-1}{t}, \ctx{st}{P, t-1}{t-1, t})$.
                Then it applies the next-message function homomorphically, generating
                $(\ctx{p}{t}{t}, \ctx{st}{P, t}{t})\leftarrow\QEval(pk_t, \cP_t, \ctx{v}{t}{t}, \ctx{st}{P, t-1}{t})$.
                It sends $\ctx{p}{t}{t}$ back to the verifier.
        \end{enumerate}
        \item $\Vblind$ decrypts the prover's final message by $p_T\leftarrow\QDec(sk_T, \ctx{p}{T}{T})$.
            It then computes its output $o\leftarrow\cV_{out}(p_T, st_{V, T})$.
    \end{enumerate}
\end{protocol}

By the correctness of $\QHE$, the completeness error of $\Piblind$ is negligibly close to that of $\Pi$.
In particular, note that the level parameter $L$ is sufficient for the honest prover which has a bounded complexity.
For the soundness property, we show the following lemma, which implies that $\Piblind$ preserves the soundness of $\Pi_0$.

\begin{thm} \label{thm:compiler-errors}
    For all cheating $\BQP$ provers $\Pblindstar$, there exists a cheating $\BQP$ prover $\Pstar$ s.t. for all $\lambda$ and inputs $x\in\zo^{\poly(\lambda)}$, the output distributions of $(\Pblindstar, \Vblind(x))$ and $(\Pstar, V)(x)$ are identical.
\end{thm}
\begin{proof}
    We define $\Pstar$ as follows.
    
    For the first rounds, it generates
    $(pk_1, sk_1)\leftarrow\QGen(1^\lambda, 1^L)$, then produces the encryptions
    $\ctx{x}{}{1}\leftarrow\QEnc(pk_1, x)$ and $\ctx{v}{1}{1}\leftarrow\QEnc(pk_1, v_1)$.
    It then runs $(\ctx{p}{1}{1}, \ctx{st}{P, 1}{1})\leftarrow\cPblind{1}(1^\lambda, \ctx{v}{1}{1}, \ctx{x}{}{1})$.
    Finally, it decrypts $p_1\leftarrow\QDec(sk_1, \ctx{p}{1}{1})$ and sends it back to the verifier,
    and keeps $\ctx{st}{P, 1}{1}$ and $sk_1$.

    For the other rounds, it generates
    $(pk_t, sk_t)\leftarrow\QGen(1^\lambda, 1^L)$, and produces ciphertexts
    $\ctx{v}{t}{t}\leftarrow\QEnc(pk_t, v_t))$ and $\ctx{sk}{t-1}{t}\leftarrow\QEnc(pk_t, sk_{t-1})$.
    It then runs $(\ctx{p}{t}{t}, \ctx{st}{P, t}{t})\leftarrow\cPblind{t}(\ctx{v}{t}{t},$ \ $\ctx{sk}{t-1}{t}, \ctx{st}{P, t-1}{t-1})$.
    Finally, it decrypts $p_t\leftarrow\QDec(sk_t, \ctx{p}{t}{1})$ and sends it back to the verifier,
    and keeps $\ctx{st}{P, t}{t}$ and $sk_t$.
        
    By construction, the experiments $(\Pblindstar, \Vblind(x))$ and $(\Pstar, V)(x)$ are identical.
\end{proof}

Finally, we show the blindness of $\Piblind$ through a standard hybrid argument where the $sk_i$'s are ``erased" one by one, starting from $sk_T$.
Once $sk_1$ is eventually erased, $\QEnc(pk_1, x)$ and $\QEnc(pk_1, 0)$ become indistinguishable due to the IND-CPA security of $\QEnc$, and we complete the proof.
We now fill in the mathmatical details of this argument.

\begin{thm} \label{thm:compiler-blindness}
    Under the QLWE assumption with polynomial modulus, $\Piblind$ is blind.
\end{thm}
\begin{proof}
    We show that for all cheating $\BQP$ provers $P^*$, $\lambda\in\bbN$, $x\in\zo^n$,
    $P^*$ cannot distinguish $(P^*, \Vblind(x))(1^\lambda)$ from $(P^*, \Vblind(0^n))(1^\lambda)$ with noticeable probability in $\lambda$.
    We use a hybrid argument; let $\Hyb_{T+1}^x=(P^*, \Vblind(x))(1^\lambda)$ and $\Hyb_{T+1}^0=(P^*, \Vblind(0^n))(1^\lambda)$.
    For $2\leq t<T+1$, define $\Hyb_t^x$ to be the same as $\Hyb_{t+1}^x$,
    except when $\Vblind$ should send $\ctx{v}{t}{t}$ and $\ctx{sk}{t-1}{t}$, it instead sends encryptions of $0$ under $pk_t$.
    We define $\Hyb_1^x$ to be the same as $\Hyb_2^x$ except the verifier sends encryptions of $0$ under $pk_1$ in place of $\ctx{x}{}{1}$ and $\ctx{v}{1}{1}$.
    We define $\Hyb_t^0$ similarly. Note that $\Hyb_1^x$ and $\Hyb_1^0$ are identical.

    For all $t$, from the perspective of the prover,
    as it receives no information on $sk_t$,
    $\Hyb_{t+1}^x$ is computationally indistinguishable from $\Hyb_t^x$ due to the CPA security of $\QHE$ under $pk_t$.
    By a standard hybrid argument, we observe that $\Hyb_1^x$ is computationally indistinguishable with $\Hyb_{T+1}^x$.
    We use the same argument for the computational indistinguishability between $\Hyb_1^0$ and $\Hyb_{T+1}^0$.
    We conclude that $P^*$ cannot distinguish between $\Hyb_{T+1}^x$ and $\Hyb_{T+1}^0$,
    therefore $\Piblind$ is blind.
\end{proof}

Applying our compiler to the parallel repetition of Mahadev's protocol for $\BQP$ from \cite{arXiv:ChiaChungYam19, arXiv:AlaChiHun19} and our $\QPIP_0$ protocol $\PiSampZ$ from \Cref{proto:QPIP0samp} for $\SampBQP$ yields the first constant-round blind $\QPIP_0$ protocol for $\BQP$ and $\SampBQP$, respectively.

\begin{thm}
    Under the QLWE assumption, there exists a blind, four-message $\QPIP_0$ protocol for all languages in $\BQP$ with negligible completeness and soundness errors.
\end{thm}

\begin{thm}
        Under the QLWE assumption, there exists a blind, four-message $\QPIP_0$ protocol for all sampling problems in $\SampBQP$ with negligible completeness error and computational soundness.
\end{thm}

\section*{Acknowledgments}

The authors would like to thank Tomoyuki Morimae for his valuable feedback that helped improve the paper and for pointing out the related works \cite{takeuchi2018verification, hayashi2015verifiable}.
We are also thankful to anonymous reviewers for various useful comments.

\bibliographystyle{plain}
\bibliography{refs}

\begin{appendices}
\newpage

\appendix

\section{Construction of the $\QPIP_1$ Protocol for $\SampBQP$}
\label{sec:sampbqp}

\subsection{Reducing $\SampBQP$ to the X-Z Local Hamiltonian} \label{sec:LHXZ}

Recall the definition of the history state which serves as a transcript of the circuit evolution~\cite{kitaev2002classical}:

\begin{dfn}[History-state]
    \label{dfn:groundstate}    
    Given any quantum circuit $C=U_T\ldots U_1$ of $T$ elementary gates and input $x\in\{0,1\}^n$, the \emph{history}-state $\histpsi{C(x)}$ is defined by
    \begin{equation}
        \histpsi{C(x)} \equiv \frac{1}{\sqrt{T}}\sum_{t=0}^{T-1}U_t\ldots U_1\ket{x,0}\otimes\ket{\hat{t}},
    \end{equation}
    where the first register of $n$-qubit refers to the input, the second of $m$-qubit refers to the work space ($\mathrm{poly}(n)$, w.l.o.g, $\leq T$ size) which is initialized to $\ket{0}$, and the last refers to the clock space which encodes the time information. Note that $\hat{t}$ could be some representation of $t$=0,..., $T-1$.
\end{dfn}

We define the X-Z local Hamiltonian of interest as follows:

\begin{dfn}[$k$-local X-Z Hamiltonian] For any $n$-qubit system, the set of $k$-local X-Z terms, denoted by $\LHXZ{k}$, contains Hermitian matrices that apply non-trivially on at most $k$ qubits as a product of Pauli $X$ and $Z$ terms. Namely,
\begin{equation}
  \LHXZ{k} = \left\{h_1 \otimes h_2 \otimes \ldots \otimes h_n: \forall i \in [n], h_i \in \{I, X, Z\}, \text{and}, \abs{\set{i: h_i=X \text{ or }Z}} \leq k. \right \}.
\end{equation}
A $k$-local X-Z Hamiltonian $H$ is a linear combination of terms from $\LHXZ{k}$. Namely,
\begin{equation}
  H = \sum_i \alpha_{i} H_i,  \quad \forall i, \alpha_i \in R,  H_i \in \LHXZ{k}.
\end{equation}
\end{dfn}

As our starting point, we will include the existing result of constructing X-Z local Hamitonians for general $\BQP$ computation.
First, we note the fact Toffoli and Hadamard gates form a universal gate set for quantum computation~\cite{Shi03, quant-ph/0301040}.
It is easy to see that both Toffoli and Hadamard gates can be represented as linear combinations of terms from $\LHXZ{}$. For example,
 \begin{equation}
     \mathrm{(Hadamard)} \quad H \equiv \frac{1}{\sqrt{2}} \begin{pmatrix}1&1\\1&-1\end{pmatrix} = \frac{1}{\sqrt{2}}\left (X+Z\right).
 \end{equation}
 The Toffoli gate maps bits $(a,b,c)$ to $(a,b, c \oplus (a \text{ and } b))$, which can be decomposed as $\ket{11}\bra{11}\otimes X+(I-\ket{11}\bra{11})\otimes I$ where $\ket{11}\bra{11}=\frac{1}{4}(I\otimes I+Z\otimes Z-I\otimes Z-Z\otimes I)$.
 
We will follow the unary clock register design from Kitaev's original 5-local Hamiltonian construction~\cite{kitaev2002classical}. Namely, valid unary clock states ($T$-qubit) are $\ket{00\ldots0}$, $\ket{10\dots0}$, $\ket{110\ldots0}$, etc, which span the ground energy space of the following local Hamiltonian:
\begin{equation}
    \Hclock= \sum_{t=1}^{T-1}\proj{01}_{t,t+1},
\end{equation}
where $\proj{01}_{t,t+1}$ stands for a projection on the $t$th and $(t+1)$th qubit in the clock register.
It is observed in~\cite{PhysRevA.78.012352} that $\Hclock$ can be reformulated as a linear combination of terms from $\LHXZ{}$ as follows,
\begin{equation} \label{eqn:Hclock}
   \Hclock= \frac{1}{4}(Z_1 - Z_T) + \frac{1}{4}\sum_{t=1}^{T-1}(I-Z_t\otimes Z_{t+1}),
\end{equation}
where $Z_t$ refers to Pauli $Z$ operated on the $t$th qubit in the clock register.

One can achieve so similarly for $\Hin$ and $\Hprop$. For the input condition, we want to make sure $x=(x_1, \ldots, x_n) \in \{0,1\}^n$ is in the input space and the workspace is initialized to $\ket{0}$ for all qubits at the time $0$. Namely, one can set $\Hin$ to be
\begin{equation}
    \Hin = \sum_{i=1}^n(I- \proj{x_i}_i)\otimes \proj{0}_1 + \sum_{i=1}^m \proj{1}_{n+i} \otimes \proj{0}_1,
\end{equation}
where the last part $\proj{0}_1$ applies on to the first qubit in the clock register. One can rewrite $\Hin$ as
\begin{equation}\label{eqn:Hin}
 \Hin=\frac{1}{4}\sum_{i=1}^n(I-(-1)^{x_i}Z_i)\otimes(I+Z_1) + \frac{1}{4} \sum_{i=1}^m (I - Z_{n+i}) \otimes (I + Z_1).
\end{equation}

For the propagation of the quantum state through the circuit, one uses the $\Hprop$ as follows:
\begin{equation} \label{eqn:Hprop}
    \Hprop=\sum_{t=1}^T \Hprop^t,
\end{equation}
where
\begin{equation}
    \Hprop^t=\frac{1}{2}I\otimes\proj{\widehat{t}}
    +\frac{1}{2}I\otimes\proj{\widehat{t-1}}
    -\frac{1}{2} U_t\otimes\ket{\widehat{t}}\bra{\widehat{t-1}}
    -\frac{1}{2}U_t^\dagger\otimes\ket{\widehat{t-1}}\bra{\widehat{t}}.
\end{equation}
Note that $\ket{\widehat{t}}\bra{\widehat{t-1}}=\ket{110}\bra{100}_{(t-1,t,t+1)}$ and similarly for $\ket{\widehat{t-1}}\bra{\widehat{t}}$.
Note that $U_t^\dagger=U_t$, since our gates are either Hadamard or Toffoli. It is observed in~\cite{PhysRevA.78.012352} that
\begin{equation}
   \Hprop^t=\frac{I}{4}\otimes(I-Z_{t-1})\otimes (I+Z_{t+1})-\frac{U_t}{4}\otimes(I-Z_{t-1})\otimes X_t\otimes (I+Z_{t+1}), \forall t=2, \ldots, T-1,
\end{equation}
and
\begin{eqnarray}
  \Hprop^1 &= & \frac{1}{2}(I+Z_2)-U_1\otimes\frac{1}{2}(X_1+X_1\otimes Z_2), \\
  \Hprop^T &= & \frac{1}{2}(I-Z_{t-1})-U_T\otimes\frac{1}{2}(X_T-Z_{T-1}\otimes X_T).
\end{eqnarray}
Combining with the fact that each $U_t$ can be written as a linear combination of terms from $\LHXZ{}$, we conclude that $\Hclock$, $\Hin$, $\Hprop$ are 6-local X-Z Hamiltonian.

We will employ the perturbation technique to amplify the spectral gap of $\Hclock + \Hin + \Hprop$.
Let $\ground{H}$ denote the ground energy of any Hamiltonian $H$.
The projection lemma from \cite{kempe_kitaev_regev_2006} approximates $\ground{H_1 + H_2}$ in terms of $\ground{H_1\big|_{\ker H_2}}$, where $\ker H_2$ denotes the \emph{kernel} space of $H_2$.

\begin{lem}[Lemma 1 in \cite{kempe_kitaev_regev_2006}]
    \label{thm:proj1}
    Let $H=H_1+H_2$ be the sum of two Hamiltonians operating on Hilbert space $\cH=\cS+\cS^\bot$.
    The Hamiltonian $H_2$ is such that $\cS$ is a zero eigenspace and the eigenvectors in $\cS^\bot$ have eigenvalues at least $J>2\norm{H_1}$. Then,
    $$\lambda\left(H_1\big|_\cS\right)-\frac{\norm{H_1}^2}{J-2\norm{H_1}}\leq\lambda(H)\leq\lambda\left(H_1\big|_\cS\right).$$
\end{lem}

We will use the following simple reformulation instead.

\begin{lem}
    \label{lem:projection}
    Let $H_1, H_2$ be local Hamiltonians where $H_2\geq0$. Let $K=\ker H_2$ and
    $$J=\frac{8\norm{H_1}^2+ 2\norm{H_1}}{\lambda\left(H_2\big|_{K^\bot}\right)},$$
    then we have
    $$\lambda(H_1+JH_2)\geq\lambda\left(H_1\big|_K\right)-\frac{1}{8}.$$
\end{lem}
\begin{proof}
    Apply \Cref{thm:proj1} to $H=H_1+JH_2$. Note that the least non zero eigenvalue of $JH_2$ is greater than $2\norm{H_1}$.
\end{proof}

\begin{thm}
    \label{thm:LHReduction}
    Given any quantum circuit $C = U_T \ldots U_1$ of  $T$ elementary gates and input $x \in \{0,1\}^n$, one can construct a 6-local X-Z Hamiltonian $H_{C(x)}$ in polynomial time such that
    \begin{enumerate}
        \item[(1)] $H_{C(x)} = \sum_i \alpha_i H_i$ where
        each $H_i \in \LHXZ{6}$ and $|\alpha_i| \in O(T^9)$. Moreover, there are at most $O(T)$ non-zero terms.
        \item[(2)] $H_{C(x)}$ has $\histpsi{C(x)}$ as the unique ground state with eigenvalue $0$ and has a spectral gap at least $\frac{3}{4}$. Namely,  for any state $\ket\phi$ that is orthogonal to $\histpsi{C(x)}$, we have $\braket{\phi|H_{C(x)}|\phi}\geq \frac{3}{4}$.
    \end{enumerate}
\end{thm}

\begin{proof}
We will use the above construction $\Hclock$ (\cref{eqn:Hclock}), $\Hin$ (\cref{eqn:Hin}), $\Hprop$ (\cref{eqn:Hprop}) as our starting point, which are already 6-local X-Z Hamiltonian constructable in polynomial time. However, $H_{\mathrm{old}}=\Hin + \Hclock + \Hprop$ does not have the desired spectral gap. To that end, our construction will be a weighted sum of $\Hin$, $\Hclock$, and $\Hprop$ as follows,
\begin{equation}
    H_{\mathrm{new}}= \Hin + \Jclock \Hclock + \Jprop \Hprop,
\end{equation}
where $\Jclock$ and $\Jprop$ will be obtained using \Cref{lem:projection}.

Let $\Kin=\ker \Hin$, $\Kclock=\ker \Hclock$, and $\Kprop=\ker \Hprop$. It is known from e.g.,~\cite{kitaev2002classical}, that
\[
   \Kin \cap \Kclock \cap \Kprop = \spn\set{\histpsi{C(x)}}.
\]
Thus $\histpsi{C(x)}$ remains in the ground space of $H_{\mathrm{new}}$. Let $S$ denote its orthogonal space. Namely, $S=(\spn\set{\ket{\psi_{C(x)}}})^\bot$.
Denote by $\Hin\big|_S$ the restriction of $\Hin$ on space $S$ and similarly for others.

Consider $\Hin + \Jclock\Hclock$ first. According to \cref{lem:projection}, by choosing
\[
  \Jclock = \frac{8\norm{\Hin\big|_S}^2 + 2 \norm{\Hin\big|_S}}{\ground{\Hclock\big|_{S\cap \Kclock^\bot}}} = O(T^2),
\]
where we use the fact $\norm{\Hin|_S}\leq T$ and $\ground{\Hclock\big|_{S\cap \Kclock^\bot}}\geq \ground{\Hclock\big|_{\Kclock^\bot}}=1$,
we have
\[
 \ground{\Hin\big|_S + \Jclock\Hclock\big|_S}\geq \ground{\Hin\big|_{S\cap \Kclock}} - \frac{1}{8}.
\]
Consider further adding $\Jprop\Hprop$ term.  By choosing
\[
 \Jprop= \frac{O(\norm{\Hin\big|_S + \Jclock\Hclock\big|_S}^2)}{\ground{\Hprop\big|_{S \cap \Kprop^\bot}}}= O(T^8),
\]
where we use the fact $\ground{\Hprop\big|_{S\cap \Kprop^\bot}}\geq \ground{\Hprop\big|_{\Kprop^\bot}}=\Omega(T^{-2})$~\cite{kitaev2002classical}, we have  
\[
 \ground{\Hin\big|_S + \Jclock\Hclock\big|_S + \Jprop\Hprop\big|_S} \geq \ground{\Hin\big|_{S \cap \Kclock \cap \Kprop}} - \frac{1}{4}.
\]
A simple observation here is that $S \cap \Kclock \cap \Kprop$ is the span of history states with different inputs or different initialization of the work space. Namely, $\ground{\Hin\big|_{S \cap \Kclock \cap \Kprop}} \geq 1$. Thus,
\[
  \ground{(\Hin + \Jclock\Hclock + \Jprop\Hprop)\big|_S} \geq 1- \frac{1}{4}=\frac{3}{4}.
\]
Given that $\histpsi{C(x)}$ is the ground state of $H_{\mathrm{new}}$ with eigenvalue 0, and any orthogonal state to $\histpsi{C(x)}$ has eigenvalue at least $\ground{H_{\mathrm{new}}}\geq 3/4$, the spectral gap of $H_{\mathrm{new}}$ at at least $3/4$.

Note that $H_{\mathrm{new}}$ is a 6-local X-Z Hamiltonian by construction. It suffices to check the bound of $\abs{\alpha_i}$ and the number of terms. The former is one more than the order of $\Jprop$ since each $\Hprop^t$ contributes $\frac{1}{4}$ to the $I$ term, creating an extra factor of $T$. The latter is by counting the number of terms from $\Hin, \Hclock, \Hprop$, each of which is bounded by $O(T)$.
\end{proof}

\noindent \textbf{Remark.} We believe that the specific parameter dependence above can be tightened by a more careful analysis. However, as our focus is on the feasibility, we keep the above slightly loose analysis which might be more intuitive.

\subsection{Delegation Protocol for $\QPIP_1$ client}
\label{sec:qpip1}
In this subsection, we construct a one-message $\QPIP_1$ delegation protocol for $\SampBQP$. By definition of $\QPIP_1$, we assume the client has limited quantum power, e.g., performing single qubit $X$ or $Z$ measurement one by one.
Intuitively, one should expect the one-message from the server to the client is something like the history state so that the client can measure to sample.  

At a high-level, the design of such protocols should consist of at least two components: (1) the first component should test whether the message is indeed a valid history state; (2) the second component should simulate the last step of any $\SampBQP$ computation by measuring the final state in the computational basis.

Our construction of X-Z local Hamiltonian $H$ from \Cref{thm:LHReduction} will help serve the first purpose.
In particular, we adopt a variant of the energy verification protocol for local Hamiltonian (e.g., ~\cite{mf16, PhysRevA.93.022326}) to certify the energy of $H$ with only $X$ or $Z$ measurements.
Moreover, because of the large spectral gap, when the energy is small, the underlying state must also be close to the history state.
Precisely, consider the following protocol $\cVGS$:

\begin{protocol}{Energy verification for X-Z local Hamiltonian $\cVGS$} \label{AlgGroundStateCheck}
Given a $k$-local X-Z Hamiltonian
$H=\sum_i \alpha_{i} H_i$ (i.e., $\forall i, H_i \in \LHXZ{k}$) and any state $\ket{\phi}$.

\begin{itemize}
\item Let $p_i= \abs{\alpha_i}/\sum_i \abs{\alpha_i}$ for each $i$. Sample $i^*$ according to $p_{i^*}$.
\item Pick $H_{i^*}$ which acts non-trivially on at most $k$ qubits of $\ket{\phi}$. Measure the corresponding single-qubit Pauli X or Z operator.
Record the list of the results $x_j=\pm 1$ for $j=1, \ldots k$.
\item Let $r=x_1x_2\cdots x_k$. The protocol \emph{accepts} if $r$ and $\alpha_{i^*}$ have different signs, i.e., $\sgn(\alpha_{i^*})r=-1$. Otherwise, the protocol \emph{rejects}.
\end{itemize}
\end{protocol}

\begin{lem}[\cite{PhysRevA.93.022326}]
    \label{thm:HamCheck}
    For any $k$-local X-Z Hamiltonian $H=\sum_i \alpha_{i} H_i$ and any state $\ket{\phi}$,
    the protocol $\cVGS$ in Protocol~\ref{AlgGroundStateCheck} accepts with
    probability
\begin{equation}
 \mathrm{Prob}[ \cVGS \text{ accepts } \ket{\psi}] = \frac{1}{2} - \frac{1}{2 \sum_i \abs{\alpha_i}}\braket{\phi|H|\phi}.
\end{equation}
\end{lem}

\begin{theorem} \label{thm:HamCheckClose}
Given any quantum circuit $C$ and input $x$, consider using $H_{C(x)}$ from \Cref{thm:LHReduction} in Protocol~\ref{AlgGroundStateCheck} ($\cVGS$).
For any state $\rho$, and $0< \epsilon < 1$, if $\cVGS$ accepts $\rho$ with probability
\[
 \mathrm{Prob}[\cVGS \text{ accepts } \rho] \geq \frac{1}{2} - \frac{\epsilon}{2 \sum_i \abs{\alpha_i}},
\]
then the trace distance between $\rho$ and $\histpsi{C(x)}$ is at most $\frac{2}{\sqrt{3}}\sqrt{\epsilon}$.
\end{theorem}

\begin{prf} Consider the pure state case $\rho=\proj{\phi}$ first. By \Cref{thm:HamCheck} and our assumption, we have $\braket{\phi|H|\phi} \leq \epsilon$.
Decompose $\ket{\phi}= \alpha \histpsi{C(x)} + \beta \histpsi{C(x)}^\bot$. Note that $\histpsi{C(x)}$ is an eigenvector $H_{C(x)}$ of eigenvalue 0 and all other eigenvalues are at least $3/4$. Thus, we have $\abs{\alpha}^2 \geq 1-\frac{4}{3}\epsilon$. Thus,
\[
   \norm{\proj{\psi^{\mathrm{hist}}_{C(x)}}- \proj{\phi}}_{\tr} = \sqrt{1- |\braket{\psi^{\mathrm{hist}}_{C(x)}|\phi}|^2}
   \leq \frac{2}{\sqrt{3}}\sqrt{\epsilon}.
\]
For any mixed state $\rho=\sum_i p_i \proj{\phi_i}$, by the triangle inequality, we have
\[
    \norm{\proj{\psi^{\mathrm{hist}}_{C(x)}}- \rho}_{\tr} \leq \sum_i p_i    \norm{\proj{\psi^{\mathrm{hist}}_{C(x)}}- \proj{\phi_i}}_{\tr} \leq \sum_i p_i \frac{2}{\sqrt{3}} \sqrt{\epsilon} =\frac{2}{\sqrt{3}}\sqrt{\epsilon}.
\]
\end{prf}

To serve the second purpose, one needs to combine the test and the output on multiple copies of the history states, where we construct the following cut-and-choose protocol.  
The challenge comes from the fact that a cheating prover might send something rather than copies of the history state.
In particular, the prover can entangle between different copies in order to cheat.
In the case of certifying a $\BQP$ computation, the goal is to verify the ground energy of any local Hamiltonian.
A cheating strategy with potential entanglement won't create any witness state with an energy lower than the actual ground energy.
Thus, this kind of attack won't work for $\BQP$ computation.

However, in the context of $\SampBQP$, one needs to certify the ground energy (known to be zero in this case) and to output a good copy.
While the statistical test as before can be used to certify the ground energy,
it has less control on the shape of the output copy.
In fact, the prover can always prepare a bad copy along with many good copies as a plain attack.
This attack will succeed when the bad copy is chosen to output, the probability of which is non-negligible in terms of the total number of copies assuming some symmetry of the protocol.
The potential entanglement among different copies could further complicate the analysis.
We employ the quantum \emph{de Finetti}'s theorem to address this technical challenge.
Specifically, given any permutation-invariant $k$-register state, it is known that the reduced state on many subsets of $k$-register will be close to a separable state.
This helps establish some sort of independence between different copies in the analysis.
To serve our purpose, we adopt the following version of quantum de Finetti's theorem from~\cite{Brandao2017} where the error depends nicely  on the number of qubits, rather than the dimension of quantum systems.

\begin{thm}[\cite{Brandao2017}]
    \label{deFinetti}
    Let $\rho^{A_1\ldots A_k}$ be a permutation-invariant state on registers $A_1,\ldots,A_k$ where each register contains $s$ qubits.
    For any $0\leq l\leq k$,  there exists states $\set{\rho_i}$ and $\set{p_i}\subset\bbR$ such that
    $$\max_{\Lambda_1,\ldots,\Lambda_l}
    \norm{(\Lambda_1\otimes\ldots\otimes\Lambda_l)\left(\rho^{A_1\ldots A_l}-\sum_ip_i\rho_i^{A_1}\otimes\ldots\otimes\rho_i^{A_l}\right)}_1
    \leq\sqrt{\frac{2l^2s}{k-l}}$$
    where $\Lambda_i$ are quantum-classical channels.
\end{thm}

\begin{protocol}{$\QPIP_1$ protocol $\PiSamp=(\PSamp, \VSamp)$ for the $\SampBQP$ problem $(D_x)_{x\in\set{0, 1}^*}$}
	\label{ProtoQPIP1}

	Inputs:
	\begin{itemize}
		\item Security parameter $1^\lambda$ where $\lambda\in\bbN$
		\item Accuracy parameter $1^{1/\eps}$ where $\eps\in(0, 1)$
		\item Classical input $x\in\zo^n$ to the $\SampBQP$ instance
	\end{itemize}

	Ingredients:
	\begin{itemize}
		\item Let $C$ be a quantum circuit consisting of only Hadamard and Toffoli gates, which on input $x$ samples from some $C_x$ such that $\norm{C_x-D_x}_{\mathrm{TV}}\leq\eps$.
		\item Let $T$ be the number of gates in $C$.
		\item Let $C'$ be $C$ padded with $\frac{6T}{\eps}$ identity gates at the end.
		\item Let $H_{C'(x)}=\sum_i \alpha_i H_i$ be the X-Z local Hamiltonian instance associated with the computation of $C'(x)$ from \Cref{thm:LHReduction}.
		\item Let $\histpsi{{C'(x)}}$ be the ground state of $H_{C'(x)}$.
		\item Let $s\leq 3T+\frac{6T}{\eps}$ be the number of qubits in $\histpsi{C'(x)}$.
		\item Let $\cVGS$ be the verification algorithm for $H_{C'(x)}$ as defined in \Cref{AlgGroundStateCheck}.
	\end{itemize}

	Protocol:
	\begin{enumerate}
		\item\label{step:qpip1-state-gen} The honest prover prepares $M=649\frac{T^{41}\lambda^2}{\eps^{51}}$ copies of $\histpsi{C'(x)}$ and sends all of them to the verifier qubit-by-qubit.
		\item\label{step:qpip1-verify} The verifier samples $I\subset[M]$ s.t. $\abs{I}=m$ where $m=\frac{T^{20}\lambda}{\eps^{24}}$, and $k\xleftarrow{\$}[M]\setminus I$.
			For $i$ from $1$ to $M$, it chooses what to do to the $i$-th copy, $\rho_i$, as follows:
		\begin{enumerate}
			\item If $i\in I$, run $y_i\leftarrow\cV_{\GS}(\rho_i)$.
			\item Else if $i=k$, measure the data register 
			of $\rho_i$ under the computational (i.e., $Z$) basis and save the outcome as $z$.
			\item Else, discard $\rho_i$.
		\end{enumerate}
			Let $Y=\sum_{i\in I} y_i$. If $Y>\frac{m}{2}-\kappa m$ where $\kappa=\frac{1}{192}\frac{\eps^2}{\sum_i\abs{\alpha_i}}$, then the verifier outputs $(\Acc, z)$.
			Else, it outputs $(\Rej, \bot)$.
	\end{enumerate}
\end{protocol}

Note that $\VSamp$ only needs to apply $X$ and $Z$ measurements, and is classical otherwise.
It is simple to check that $\VSamp$ runs in $\poly(n, \lambda, \frac{1}{\epsilon})$ time.
We now show its completeness and soundness.

\begin{prf} [\Cref{QPIP1thm}]
	For completeness, notice $y_i$ are i.i.d. Bernoulli trials with success probability $\frac{1}{2}$.
	So we can apply Chernoff bound to get
	$$\Prob{Y>\frac{m}{2}-\eps m}\leq\negl(\lambda).$$
	So $\VSamp$ has negligible probability to reject.
	
	Now we show soundness.
	Suppose $\PSampstar$ is a cheating prover that, on inputs $x, 1^{1/\eps}, 1^\lambda$, sends some $\sigma$ to the verifier.

	The first step of our analysis is to use de Finetti's theorem.
	Randomly picking $m+1$ out of $M$ registers is equivalent to first applying a random permutation then taking the first $m+1$ registers.
	A random permutation, in turn, is a classical mix over all possible permutations:
	$$\sigma'=\frac{1}{\abs{\Sym(M)}}\sum_{\Pi\in\Sym(M)}\Pi\sigma\Pi^\dagger.$$
	It is simple to check that $\sigma'$ is permutation-invariant:
	fix $\tilde{\Pi}\in\Sym(M)$, then
	$$\tilde{\Pi}\sigma'\tilde{\Pi}^\dagger
	=\frac{1}{\abs{\Sym(M)}}\sum_{\Pi\in\Sym(M)}\tilde{\Pi}\Pi\sigma\Pi^\dagger\tilde{\Pi}^\dagger
	=\frac{1}{\abs{\Sym(M)}}\sum_{\hat{\Pi}\in\Sym(M)}\hat{\Pi}\sigma\hat{\Pi}^\dagger
	=\sigma',$$
	where the second equality is by relabeling $\tilde{\Pi}\Pi=\hat{\Pi}$.

	Now we apply de Finetti's theorem (\Cref{deFinetti}) to approximate our measurement result on $\sigma'$ with the measurement result of a classical mix over $(m+1)$-identical copies of some states $\tau_j$.
	That is, there exists some $\rho=\sum_j w_j\tau_j^{\otimes m+1}$
	such that for all quantum-classical channels $\Lambda_i$ acting on $s$-qubits:
	$$\norm{\Lambda_1\otimes\ldots\otimes\Lambda_{m+1}(\sigma'_{\leq m+1}-\rho)}_1
	\leq\sqrt{\frac{2m^2s}{M-m}}
	\leq\frac{\eps}{6},$$
	where $\sigma'_{\leq m+1}$ is the first $m+1$ registers of $\sigma'$.
	In our context, for $1\leq i\leq m$, $\Lambda_i$ corresponds to measurements chosen by $\cVGS$.
	$\Lambda_{m+1}$ measures the data register under the standard basis.
	
	We now analyze the verifier's output for each $\tau_j$ in $\rho=\sum_j w_j\tau_j^{\otimes m+1}$. Let $(d_j, z_j)$ be the verifier's output corresponding to $\tau_j$, and define $z_{j, ideal}$ accordingly.
	We claim that $\norm{(d_j, z_j)-(d_j, z_{j, ideal})}_{\mathrm{TV}}<\frac{2\eps}{6}$.
	Let $p_j$ be the probability that $\cVGS(\tau_j)=\Acc$, and consider the following two cases.

	First suppose $p_j<\frac{1}{2}-2\kappa$,
	then a standard Chernoff bound argument can be applied show that this case has negligible acceptance probabilities.
	As a result, $\norm{(d_j, z_j)-(d_j, z_{j, ideal})}_{\mathrm{TV}}=\negl(\lambda)$.

	Now suppose $p_j\geq\frac{1}{2}-2\kappa$.
	By \Cref{thm:HamCheckClose} we have
	$$\norm{\tau_j - \ket{\psi^{\mathrm{hist}}_{C'(x)}}\bra{\psi^{\mathrm{hist}}_{C'(x)}}}_{\tr}\leq\frac{2}{\sqrt{3}}\sqrt{2\kappa\cdot2\sum_i\abs{\alpha_i}}=\frac{\eps}{6}.$$
	Observe that when $d_j=\Acc$, $z_j$ is the measurement results on $\tau_j$'s data register,
	which is $\frac{\eps}{6}$-close to that of $\ket{\psi_{C'(x)}}$.
	As the clock register is traced out from $\ket{\psi_{C'(x)}}$, the data register has at least $1-\frac{\eps}{6}$ probability to contain $C(x)$ due to the padding.
	So $z_j$ is $\frac{2\eps}{6}$-close to $z_{ideal}$ when $d=\Acc$,
	which implies that $(d_j, z_j)$ is $\frac{2\eps}{6}$-close to $(d_j, z_{j, ideal})$.

    Finally, let $(d_\rho, z_\rho)$ denote the verifier's output distribution corresponding to $\rho=\sum_j w_j\tau_j^{\otimes m+1}$,
	and define $z_{\rho, ideal}$ accordingly.
	We have $\norm{(d_\rho, z_\rho) - (d, z_{\rho, ideal})}_{\mathrm{TV}} < \frac{2\eps}{6}$ since the same is true for all components $\tau_j$.

	As $(d, z)$ is $\frac{\eps}{6}$-close to $(d_\rho, z_\rho)$ by the data processing inequality,
	using the triangle inequality we have $\norm{(d, z) - (d_\rho, z_{\rho, ideal})}_{\mathrm{TV}}\leq\frac{\eps}{2}$.
	Hence $\norm{d - d_\rho}_{\mathrm{TV}}\leq\frac{\eps}{2}$, which implies $\norm{(d_\rho, z_{\rho, ideal}) - (d, z_{ideal})}_{\mathrm{TV}}$ $\leq$ $\frac{\eps}{2}$.
	It then follows from the triangle inequality that
	$\norm{(d, z) - (d, z_{ideal})}_{\mathrm{TV}}\leq\eps$.
\end{prf}

\end{appendices}

\end{document}